\documentclass{lmcs}
\usepackage{a4}
\usepackage[parfill]{parskip}
\usepackage{amsmath,amssymb,latexsym,stmaryrd}
\usepackage{enumitem}
\usepackage{proof}
\usepackage[all]{xy}
\usepackage{tikz}
\usetikzlibrary{arrows,shapes,automata,cd}
\usetikzlibrary{calc} 
\usetikzlibrary{matrix,arrows} 
\tikzset{
	commutative diagram/.style 2 args={
		matrix of math nodes, row sep=#1,column sep=#2,
		text height=1.5ex, text depth=0.25ex},
	commutative diagram/.default={1cm}{1cm}
}
\tikzset{    
	skip loop/.style n args={3}{to path={-- ++(0,#1) -| node[pos=0.25,#2] {#3} (\tikztotarget)}},
	cross line/.style={preaction={draw=white, -, line width=6pt}}
}

\usepackage{comment}

\newcommand{\e}{\epsilon}
\newcommand{\U}{\mathcal{U}}

\newcommand{\B}{\mathcal{B}}
\newcommand{\C}{\mathcal{C}}
\newcommand{\A}{\mathcal{A}}
\newcommand{\KK}{\mathbb{K}}
\newcommand{\E}{\mathcal{E}}

\newcommand{\V}{\mathcal{V}}
\newcommand{\R}{\mathcal{R}}
\newcommand{\lng}{\mathcal L}
\newcommand{\M}{\mathcal{M}}
\newcommand{\N}{\mathcal{N}}
\newcommand{\term}{\mathbb{T}}
\newcommand{\ol}[1]{\overline{#1}}
\newcommand{\ang}[1]{\langle{#1}\rangle}
\newcommand{\md}[1]{\term_{\mathcal K}{#1}}
\newcommand{\pd}{\displaystyle\prod}

\newcommand{\T}{\term}

\newcommand{\naturals}{\mathbb{N}}
\newcommand{\prationals}{\mathbb{Q}_+}

\newcommand{\extreals}{\mathbb{R}_+ \cup \{\infty\}}

\newcommand{\cat}{\mathbf}
\newcommand{\str}[3]{(#1_{#2})_{#2\in #3}}

\usepackage{tikz}
\usetikzlibrary{calc,matrix,arrows} 
\tikzset{
	commutative diagram/.style 2 args={
		matrix of math nodes, row sep=#1,column sep=#2,
		text height=1.5ex, text depth=0.25ex},
	commutative diagram/.default={1cm}{1cm}
}
\tikzset{    
	skip loop/.style n args={3}{to path={-- ++(0,#1) -| node[pos=0.25,#2] {#3} (\tikztotarget)}},
	cross line/.style={preaction={draw=white, -, line width=6pt}}
}

\begin{document}
	
	\title{On the Axiomatizability of Quantitative Algebras} 
	
	\author[R.\ Mardare]{Radu Mardare}	
	\address{Dept.\ of Computer Science, Aalborg University, Denmark}	
	\email{mardare@cs.aau.dk}  
	
	\author[P.\ Panangaden]{Prakash Panangaden}	
	\address{School of Computer Science, McGill University, Canada}	
	\email{prakash@cs.mcgill.ca}  
	
	\author[G.\ Plotkin]{Gordon Plotkin}	
	\address{School of Informatics, University of Edinburgh, Scotland}	
	\email{gdp@inf.ed.ac.uk}  
	
	
	
	\titlecomment{An earlier version of this paper appeared as...}




\begin{abstract}
Quantitative algebras (QAs) are algebras over metric spaces defined
by quantitative equational theories as introduced by the same authors in a related paper presented at LICS 2016. These algebras
provide the mathematical foundation for metric semantics of
probabilistic, stochastic and other quantitative systems. This paper
considers the issue of axiomatizability of QAs. We investigate the
entire spectrum of types of quantitative equations that can be used
to axiomatize theories: (i) simple quantitative equations; (ii) Horn
clauses with no more than $c$ equations between variables as
hypotheses, where $c$ is a cardinal and (iii) the most general case
of Horn clauses. In each case we characterize the class of QAs and
prove variety/quasivariety theorems that extend and generalize
classical results from model theory for algebras and first-order
structures.

\end{abstract}

\maketitle 
\section{Introduction}

	In~\cite{Mardare16} we introduced the concept of a quantitative equational
theory in order to support a quantitative algebraic theory of effects and
address metric-semantics issues for probabilistic, stochastic and
quantitative theories of systems.  Probabilistic programming, in
particular, has become very important recently~\cite{Pfeffer16}, see, for
example, the web site~\cite{probprog}.  The need for semantics and
reasoning principles for such languages is important as well and recently one can witness an increased interest of the research community in this topic.
Equational reasoning is the most basic form of logical
reasoning and it is with the aim of making this available in a metric
context that we began this work.

A quantitative equational theory allows one to write equations of the form
$s =_{\epsilon}t$, where $\epsilon$ is a rational number, in order to
characterize metric structures in an algebraic context.  We developed the
analogue of universal algebras over metric spaces -- called quantitative
algebras (QAs), proved analogues of Birkhoff's completeness theorem and
showed that quantitative equations defined monads on metric spaces.  We
also presented a number of examples of interesting quantitative algebras
widely used in semantics.  We presented variants of barycentric 
algebra
\cite{Stone49} that model the space of probabilistic/subprobabilistic
distributions with either the Kantorovich, Wasserstein or total variation
metrics; the same algebras can also be used to characterize the space of
Markov processes with the Kantorovich metric.  We also gave 
a notion of
quantitative semilattice that characterizes the space of closed subsets of
an extended metric space with the Hausdorff metric.  In all these examples
we emphasized elegant axiomatizations characterizing these well-known
metric spaces. In  \cite{Bacci16} the same tools are used to provide axiomatizations for a fix-point semantics for Markov chains. Of course, some of these examples can be given by ordinary
monads, as shown in~\cite{vanBreugel07,Adamek12a}, but we are aiming to fully integrate
metric reasoning into equational reasoning.

What was left open in our previous work was what kinds of metric-algebraic
structures could be axiomatized.  This is an important issue if we want a
general theory for metric-based semantics, since we will need to understand
whether the class of systems of interest with their natural metrics can, in
fact, be axiomatized.  In the present paper, we discuss the general question
of what classes of quantitative algebras can be axiomatized by quantitative
equations, or by more general axioms like Horn clauses.

The celebrated Birkhoff variety theorem~\cite{Birkhoff35} states that a
class of algebras is equationally definable if and only if it is closed
under homomorphic images, subalgebras, and products.  Many
extensions have been proved for more general kinds of
axioms~\cite{Adamek98} and for coalgebras instead of
algebras~\cite{Adamek01,Goldblatt01}, and see
\cite{Adamek10,Barr94,Manes12} for a categorical perspective.  It is
natural to ask if there are corresponding results for quantitative
equations and quantitative algebras.  Since classical equations $s=t$
define a congruence over the algebraic structure, while quantitative
equations $s=_\e t$ define a pseudometric coherent with the algebraic
structure, the classical results do not apply directly to our case.  One
therefore needs fully to understand how metric structures behave
equationally to answer the question.  This is the challenge we take up here.

The interesting examples that we present in \cite{Mardare16} require not only axiomatizations involving quantitative equations of the form $s=_\e t$, but also conditional equations, i.e., Horn clauses involving quantitative equations. Already the simple case of Horn clauses of the form $\{x_i=_{\epsilon_i}y_i \mid i \in I\}$ as hypotheses, where $x_i,y_i$ are variables only, provides interesting examples. 
All 
this  forces us to develop some new concepts and proof techniques that are innovative in a number of ways.


Firstly, we show that considering a metric structure on top of an
algebraic structure, which implicitly requires one to replace the concept
of \textit{congruence} with a \textit{pseudometric} coherent with the
algebraic structure, is not a straightforward generalization.  Indeed, one
can always think of a congruence $\cong$ on an algebra $\A$ as to the
kernel of the pseudometric $p_{\cong}$ defined by $p_{\cong}(a,b)=0$ iff
$a\cong b$ and $p_{\cong}(a,b)=1$ otherwise.  Nevertheless, many standard
model-theoretic results about axiomatizability of algebras are particular
consequences of the discrete nature of this pseudometric.  
Many of
these results fail when one takes a
more complex pseudometric, even if its kernel remains a congruence.

Secondly, we show 
that in the case of quantitative algebras, 
quantitative equation-based axiomatizations behave very similarly to 
axiomatizations  by 
\textit{Horn clauses} involving
only quantitative equations between variables as hypotheses.  And this
remains true even when one allows functions of countable arity in the
signature.  
Horn clauses of this type are directly
connected to \textit{enriched Lawvere theories} \cite{Robinson02}.  We give a uniform treatment of all these cases by
interpreting  quantitative equations as Horn clauses with empty sets of
hypotheses.

We discover, in this context, a special class of homomorphisms that we call
\textit{$c$-reflexive homomorphisms}, for a cardinal $c$, that play a
crucial role.  These homomorphisms preserve distances on selected subsets
of cardinality less than $c$ of the metric space, i.e., any $c$-space in
the image pulls back (modulo non-expansiveness) to one in the domain.  This
concept generalizes the concept of homomorphism of quantitative algebras,
since any homomorphism of quantitative algebras is $1$-reflexive.  The
central role of  $c$-reflexive homomorphisms is
demonstrated by 
a \textit{weak universality} property, proved below. 

This result also shows that the classical canonical model construction for
classes of universal algebras is mathematically inadequate and works in the
traditional settings only because it is, coincidentally, a model isomorphic
with the more general one that we present here.  However, apart from the
classic settings (of universal algebras and congruences) the standard
construction fails to produce a 
model  isomorphic with the ``natural'' one
and consequently, it fails to reflect the weak universality properly up to
$c$-reflexive homomorphisms.

Our main result in this first part of the paper is the $c$-variety theorem
for a regular cardinal $c\leq\aleph_1$: a class of quantitative algebras
can be axiomatized by Horn clauses, each axiom having fewer than $c$
equations between variables as hypotheses, if, and only if, the class is
closed under subobjects, products and $c$-reflexive homomorphisms.  In
particular, (i) the class is a $1$-variety (closed under subobjects,
products and homomorphisms) iff it can be axiomatized by quantitative
equations; (ii) it is an $\aleph_0$-variety iff it can be axiomatized by
Horn clauses with finite sets of quantitative equations between variables
as hypotheses; and (iii) it is an $\aleph_1$-variety iff it can be
axiomatized by Horn clauses with countable sets of quantitative equations
between variables as hypotheses.  Notice that in the light of the
previously mentioned relation between congruences and pseudometrics, (i)
generalizes the original Birkhoff result for universal algebras. Without the concept of $c$-reflexivity, one can only state a quasi-variety theorem under the very strong assumption that 
reduced products always exist, as  
happens, e.g., in \cite{Weaver95}.


Thirdly, we also study the axiomatizability of classes of quantitative
algebras that admit Horn clauses as axioms, but which are not restricted to
quantitative equations between variables as hypotheses.  We prove that a
class of quantitative algebras admits an axiomatization of this type,
whenever it is closed under isomorphisms, subalgebras and what we call
\textit{subreduced products}.  
These are quantitative subalgebras
of (a special type of) products of elements in the given class; however,
while these products are always algebras, they are not always quantitative
algebras, and this is where the new concept plays its role.  This new type
of closure condition allows us to generalize the usual quasivariety theorem
of universal algebras.

Since all the isomorphisms of quantitative algebras are $c$-reflexive
homomorphisms, and since a $c$-variety is closed under subalgebras and
products, it is also closed under subreduced products, as they are
quantitative subalgebras of the product.  Hence, a $c$-variety is closed
under these operators for any regular cardinal $c>0$ and so our
quasivariety theorem extends 
the $c$-variety theorem further.  These all
are novel generalizations of the classical results. 

Last, but not least, to achieve the aforementioned results for general Horn
clauses, we had to generalize concepts and results from model theory of
first-order structures considering first-order model theory on metric
structures.  Thus, we extended to the general unrestricted case the
pioneering work in~\cite{BenYaakov08} devoted to continuous logic over
complete bounded metric spaces.  We identified the first-order counterpart
of a quantitative algebra, that we 
call a \textit{quantitative first-order
	structure}, and prove that the category of quantitative algebras is
isomorphic to the category of quantitative first-order structures.  We have
developed 
\textit{first-order equational logic} for these structures
and extended standard model theoretic results for quantitative first-order
structures.  Finally, the proof of the quasivariety theorem, which actively
involves the new concept of subreduced product, is based on a more
fundamental proof pattern that can be further used in model theory for
other types of first-order structures.  We essentially show how one can
prove a quasivariety theorem for a restricted class of first-order
structures that obey infinitary axiomatizations.

We have left behind an open question: the results regarding 
unrestricted Horn clauses have been proved under the restriction of having
only finitary functions in the algebraic signature.  This was required in order to use
standard model theoretic techniques.  We believe that a similar
result might also hold  for countable functions.

\section{Preliminaries on Quantitative Algebras}

In this section we recall some basic concepts used to define the quantitative algebras, from \cite{Mardare16}, and introduce a couple of new concepts needed in our development. 


\subsection{Quantitative Equational Theories}

Consider an \textit{algebraic similarity type} $\Omega$, which is a set containing function symbols of
finite or countable arity (we see constants as functions of arity 0). If $c$ is the arity of the function $f$ in $\Omega$, we write $f:c\in\Omega$. 
\\Given a set $X$ of variables, let $\term X$ be
the $\Omega$-\emph{term algebra} over $X$, i.e., the $\Omega$-algebra having as elements all the terms generated from the set $X$ of variables and the functions symbols in $\Omega$.
\\If $f:c\in\Omega$ and $\str{t}{i}{I}$ is an indexed family
of terms with with $|I|=c$, we write $f(\str{t}{i}{I})$ for the term obtained by applying $f$ to this family of terms in the order given by $I$.

A \emph{substitution} is a function $\sigma:X\to\term X$.  It can be
canonically extended to a homomorphism of $\Omega$-algebras $\sigma \colon \term X\to\term X$ by: 
	$$\text{for any }f:|I|\in\Omega,~~\sigma(f(\str{t}{i}{I}))=f(\sigma(t_i)_{i\in I}).$$ 
In what follows $\Sigma(X)$ denotes the set of substitutions on $\term X$.

If $\Gamma\subseteq\term X$ is a set of terms and $\sigma\in\Sigma(X)$, let
$\sigma(\Gamma)=\{\sigma(t)\mid t\in\Gamma\}$.  

We use $\V(X)$ to denote the set of \emph{indexed equations} of the form $x=_\e y$
for $x,y\in X$ and $\e\in\prationals$; similarly, we use $\V(\term X)$ to denote
the set of indexed equations of the form $t=_\e s$ for $t,s\in\term X$,
$\e\in\prationals$.  We call them \emph{quantitative equations}.

Let $\E(\T X)$ be the class of \emph{conditional quantitative equations} on
$\term X$, which are constructions of the form
$$\{s_i=_{\e_i}t_i\mid i\in I\}\vdash s=_\e t,$$ where $I$ is
a countable\footnote{Anticipating the deduction system, notice that, as usual, the hypotheses containing only variables and functions that are not present in the syntax of $\phi$ can be ignored (e.g., by involving a cut-elimination rule), and since $\phi$ can only contain a countable set of terms and functions, we can safely assume that conditional equations have a countable (possibly finite or empty) set of hypotheses.} index set, $\str{s}{i}{I}, \str{t}{i}{I}\subseteq\T X$ and
$s,t\in\T X$.  \\If $V\vdash\phi\in\E(\term X)$, we refer to the elements of
$V$ as the \emph{hypotheses} and to quantitative equation $\phi$ as the \textit{conclusion} of the conditional equation.   

When the hypotheses are only quantitative equations between variables, the quantitative conditional equation is called \textit{basic conditional equation}. These play a central role in our theory and for this reason it is useful to identify a few subclasses of them.

Given a cardinal $0< c\leq\aleph_1$, a $c$-\emph{basic conditional
	equation} on $\term X$ is a conditional quantitative equation of
the form $$\{x_i=_{\e_i}y_i\mid i\in I\}\vdash s=_\e t,$$ where $|I|<c$,
$\str{x}{i}{I}, \str{y}{i}{I}\subseteq X$ and $s,t\in\T X$.

Note that the \emph{$1$-basic conditional equations} are the conditional
equations with an empty set of hypotheses, i.e., of type $\emptyset\vdash s=_\e t$.  We call them \emph{unconditional equations} and, for simplifying notation, we often write $\vdash s=_\e t$.

The \textit{$\aleph_0$-basic conditional equations} are the conditional equations
with a finite set of hypotheses, all equating variables only.  We 
call them \textit{finitary-basic quantitative equations}.

The \textit{$\aleph_1$-basic conditional equations} are all the basic conditional
equations, hence with countable (including finite, or empty) sets of
equations between variables as hypotheses.

The conditional quantitative equations are used for reasoning, and to this end we define the concept of quantitative equational theory, which, as expected, will generalize the classical one, in the sense that $=_0$ is the classical term equality. However, for $\e\neq 0$, $=_\e$ is not an equivalence: the transitivity is replaced by a rule encoding the triangle inequality. Notice also that the rule (Arch) is infinitary and it reflects the Archimedean property of rationals. For a comprehensive study of the quantitative equational theory, see \cite{Mardare16}.  

\begin{defi}[Quantitative Equational Theory] \label{def:QEtheory}
A \emph{quantitative equational theory of type $\Omega$ over $X$} is a set $\U$ of
conditional equations on $\T X$ closed under the rules stated in Table \ref{axioms}, for
arbitrary $t,s,u \in \term X$, $\str{s}{i}{I}, \str{t}{i}{I}\subseteq\T X$,
$\e,\e'\in\prationals$, $\Gamma,\Gamma'\subseteq\V(\term X)$ and
$\phi,\psi\in\V(\term X)$.
\end{defi}

Given a quantitative equational theory $\U$ and a set $S\subseteq\U$, we say that $S$ is a set of axioms for $\U$, or $S$ axiomatizes $\U$, if $\U$ is the smallest quantitative equational theory that contains $S$.

\begin{table}
	\begin{align*} 
	\text{\textbf{(Refl)}} \quad 
	& \vdash t =_0 t \,, \\
	\text{\textbf{(Symm)}} \quad 
	& \{t=_\e s\} \vdash s=_\e t \,, \\
	\text{\textbf{(Triang)}} \quad 
	& \{t =_\e u, u =_{\e'} s \} \vdash t =_{\e+\e'} s \,, \\
	\text{\textbf{(Max)}} \quad 
	& \{t=_\e s\} \vdash t=_{\e+\e'}s \,, \text{ for all $\e'>0$} \,, \\ 
	\text{\textbf{(Arch)}} \quad 
	& \text{for }\e\geq 0,~~\{t=_{\e'}s\mid \e'>\e\} \vdash t=_\e s \,, \\
	\text{\textbf{(NExp)}} \quad
	& \text{for } f:|I|\in\Omega,~~\{t_i=_\e s_i\mid i\in I\} \vdash f(\str{t}{i}{I}) =_\e f(\str{s}{i}{I}) \,, \\
	\text{\textbf{(Subst)}} \quad
	& \text{for all $\sigma \in \Sigma(X)$, $\Gamma \vdash t =_\e s$ implies $\sigma(\Gamma) \vdash \sigma(t) =_\e \sigma(t)$} \,, \\
	\text{\textbf{(Cut)}} \quad 
	& \text{if $\Gamma \vdash \psi$ for all $\psi\in\Theta$, and $\Theta \vdash t =_\e s$, then $\Gamma \vdash t =_\e s$} \,, \\
	\text{\textbf{(Assumpt)}} \quad
	& \text{If $t =_\e s \in\Gamma$, then $\Gamma \vdash t =_\e s$} \,.
	\end{align*}
	\caption{MetaAxioms}\label{axioms}
\end{table}

A quantitative equational theory $\U$ over $\term X$ is \emph{inconsistent} if $\emptyset\vdash x=_0 y\in\U$, where $x,y\in X$ are two distinct variables. $\U$ is \emph{consistent} if it is not inconsistent.


\subsection{Quantitative  Algebras}\label{QUA}
The quantitative equational theories characterize algebras supported by metric spaces, when interpreting $s=_\e t$ as "$s$ and $t$ are at most at distance $\e$. We call them quantitative algebras.

\begin{defi}[Quantitative Algebra]\label{QA}
Given an algebraic similarity type $\Omega$, an $\Omega$-\emph{quantitative algebra} (QA) is a tuple $\A=(A,\Omega^\A,d^\A)$, where $(A,\Omega^\A)$ is an $\Omega$-algebra and $d^\A:A\times A \to \extreals$ is a metric on $A$
(possibly taking infinite values) such that all the functions in $\Omega^\A$ are \emph{non-expansive}, i.e., for any $f:|I|\in\Omega$, $\str{a}{i}{I},\str{b}{i}{I}\subseteq A$, and any $\e\geq 0$, if $d^\A(a_i,b_i)\leq\e$ for all $i\in I$, then $$d^\A(f(\str{a}{i}{I}),f(\str{b}{i}{I}))\leq\e.$$
A quantitative algebra is \textit{void} when its support is void and it is \emph{degenerate} if its support is a singleton.  
\end{defi}

As emphasized before, our intuition is that quantitative algebras generalize the concept of algebra and seen from this perspective, requiring that any function in the signature is non-expansive seems the natural way of defining the interaction between the support metric space and the algebraic structure. For the same reason the non-expensiveness must be preserved by homomorphisms.


\begin{defi}[Homomorphism of Quantitative Algebras]
Given two quantitative algebras of type $\Omega$, $\A_i=(A_i,\Omega,d^{\A_i})$, $i=1,2$, a \emph{homomorphism of quantitative algebras} is a homomorphism $h:A_1\to A_2$ of $\Omega$-algebras, which is non-expansive, i.e.,   s.t., for arbitrary $a,b \in A_1$, $$d^{\A_1}(a,b)\geq d^{\A_2}(h(a),h(b)).$$
\end{defi}
Notice that identity maps are homomorphisms and that homomorphisms are
closed under composition, hence quantitative algebras of type $\Omega$ and
their homomorphisms form a category, written $\cat{QA_\Omega}$.

\textbf{Reflexive Homomorphisms.} There are some classes of specialized homomorphisms that play a central role in describing the quasivarieties of QAs. We call them \textit{reflexive homomorphisms}. 

Hereafter we use $A\subseteq_c B$ for a cardinal $c>0$ to mean that $A$ is a subset of $B$ and $|A|<c$. Notice that $A\subseteq_{\aleph_0}B$ means that $A$ is a finite subset of $B$ and $A\subseteq_{\aleph_1}B$ means that $A$ is a countable (possible finite or void) subset of $B$.

\begin{defi}[Reflexive Homomorphism]\label{refhom}
	Given two quantitative algebras of type $\Omega$, $\A_i=(A_i,\Omega,d^{\A_i})$, $i=1,2$, a homomorphism $f:\A_1\to\A_2$ of quantitative algebras is
	\textit{$c$-reflexive}, where $c$ is a cardinal, if for any subset
	$B_2\subseteq_c A_2$ there exists a set $B_1\subseteq A_1$ such that
	$f(B_1)=B_2$ and $$\text{for any }a,b\in B_1,~~d^{\A_1}(a,b)=d^{\A_2}(f(a),f(b)).$$ 
\end{defi}

If $f:\A\to\B$ is a $c$-reflexive homomorphism, $f(A)$ is a
\emph{$c$-reflexive homomorphic image} of $A$. 

Note that any homomorphism of quantitative algebras is $1$-reflexive.
Moreover, for $c>c'$, a $c$-reflexive homomorphism is also $c'$-reflexive. 

Observe also that the restriction $f|_{B_1}:B_1\to B_2$ defined in Definition \ref{refhom} is an isometry of metric spaces. Indeed, $f|_{B_1}$ is surjective, since $f(B_1)=B_2$. It is also injective because otherwise, from $f(a)=f(b)$, we get that $d^{\A_2}(f(a),f(b))=0$, implying  $d^{\A_1}(a,b)=0$; and since $d^{\A_1}$ is a metric, we must have $a=b$.


\textbf{Quantitative Subalgebra.}
The concept of subalgebra generalizes, as expected, both the concept of $\Omega$-subalgebra and of metric subspace.

Given a quantitative algebra $\A=(A,\Omega,d^\A)$, a quantitative algebra $\B=(B,\Omega,d^\B)$ is a \textit{quantitative subalgebra} of $\A$, denoted by $\B\leq\A$, if $\B$
is an $\Omega$-subalgebra of $\A$ and for any $a,b\in B$, $d^\B(a,b)=d^\A(a,b)$.  

\textbf{Direct Products of Quantitative Algebras.}
Let $(\A_i)_{i\in I}$ be an $I$-indexed family of quantitative algebras of type
$\Omega$, where $\A_i=(A_i,\Omega,d_i)$ for all $i\in I$.  Their \emph{(direct) product}
 is the quantitative algebra $\A=(A,\Omega,d)$ such that  
\begin{itemize}
	\item $A=\pd_{i\in I}A_i$ is the direct product of the sets $A_i$, for $i\in I$;
	\item for each $f:|J|\in\Omega$ and each $a_j=(b_j^i)_{i\in I}$ for $j\in
	J$,	$$f^\A(\str{a}{j}{J})=(f^{\A_i}(\str{b^i}{j}{J}))_{i\in I};$$ 
	\item for $a=\str{a}{i}{I}$, $b=\str{b}{i}{I}$,
	$$d(a,b)=\displaystyle\sup_{i\in I}d_i(a_i,b_i).$$ 
\end{itemize}
The empty product $\pd\emptyset$ is the degenerate algebra with universe
$\{\emptyset\}$. 

The fact that this is a QA follows from the pointwise
constructions of products in both the category of $\Omega$-algebras and in
the category of metric spaces with infinite values where the product metric
is the pointwise supremum.  The non-expansiveness of the functions in the
product algebra follows from the non-expansiveness of the functions in the
components. The product quantitative algebra is written
$\pd_{i\in I}\A_i$.

Direct products have \textit{projection maps} for each $k\in I$,
$$\pi_k:\pd_{i\in I}\A_i\to\A_k,$$ defined for arbitrary $a=(a_i)_{i\in
	I}\in\pd_{i\in I}A_i$ by $\pi_k(a)=a_k$.  If none of the quantitative
algebras in the family is void, the projection maps are always surjective
homomorphisms of QAs. 


\textbf{Closure Operators.}
It is useful in what follows to define a few operators mapping classes of QAs into classes of QAs.

\begin{defi}\label{closureop}
Given a class $\mathcal K$ of quantitative algebras and a cardinal $c$, let  $\mathbb I(\mathcal K)$, $\mathbb S(\mathcal K)$, $\mathbb H_c(\mathcal K)$, $\mathbb P(\mathcal K)$ and $\mathbb V_c(\mathcal K)$ be the classes of quantitative algebras defined as follows.
\begin{itemize}
	\item $\A\in\mathbb I(\mathcal K)$ iff $\A$ is isomorphic to some
	member of $\mathcal K$; 
	\item $\A\in\mathbb S(\mathcal K)$ iff $\A$ is a quantitative
	subalgebra of some member of $\mathcal K$; 
	\item $\A\in\mathbb H_c(\mathcal K)$ iff $\A$ is the $c$-reflexive
	homomorphic image of some algebra in $\mathcal K$; in particular, we
	denote $\mathbb H_1(\mathcal K)$ simply by $\mathbb H(\mathcal K)$ since it is the closure
	under homomorphic images; 
	\item $\A\in\mathbb P(\mathcal K)$ iff $\A$ is a direct product of a
	family of elements in $\mathcal K$; 
	\item $\mathbb V_c(\mathcal K)$ is the smallest class of quantitative
	algebras containing $\mathcal K$ and closed under subalgebras,
	direct products, and $c$-reflexive homomorphic images; such a class is
	called a \textit{$c$-variety of quantitative algebras}.  In
	particular, for $c=1$ we also write $\mathbb V_1(\mathcal K)$ as  $\mathbb V(\mathcal K)$ and call it a \textit{variety}. 
\end{itemize}	
\end{defi}

For any operators
$\mathbb X,\mathbb Y\in\{\mathbb I,\mathbb S, \mathbb H_c,\mathbb P, \mathbb V_c\}$, we write $\mathbb X\mathbb Y$ for their composition; and since this composition is associative, we ignore parentheses when composing more than two operators. Furthermore, for any compositions $\mathbb X,\mathbb Y$ of these we write
$\mathbb X\subseteq\mathbb Y$ if $\mathbb X(\mathcal K)\subseteq\mathbb Y(\mathcal K)$ for any class $\mathcal K$.

The next lemma establishes a series of properties of these operators, similar to the ones on classes of universal algebras.

\begin{lem}\label{operators}
	The closure operators on classes of quantitative algebras enjoy the following properties:
	\begin{enumerate}
		\item whenever $c<c'$, $\mathbb{H}_c\subseteq \mathbb H_{c'}$; 
		\item whenever $c<c'$, if $\mathcal K$ is $\mathbb H_c$-closed, then it is $\mathbb H_{c'}$-closed; in particular, a $\mathbb H$-closed class is $\mathbb H_c$-closed for any $c$;
		\item whenever $c<c'$, if $\mathcal K$ is $c$-variety, then it is a $c'$-variety; in particular, a variety is a $c$-variety for any $c$;
		\item $\mathbb{SH}_c\subseteq\mathbb H_c\mathbb S$; in particular, $\mathbb S\mathbb H\subseteq\mathbb H\mathbb S$;
		\item $\mathbb{PH}_c\subseteq\mathbb H_c\mathbb P$; in particular, $\mathbb P\mathbb H\subseteq\mathbb H\mathbb P$;
		\item $\mathbb P\mathbb S\subseteq\mathbb S\mathbb P$;
		\item $\mathbb H_c$, $\mathbb H$, $\mathbb S$ and $\mathbb{IP}$ are idempotent;
		\item $\mathbb V_c=\mathbb H_c\mathbb {SP}$; in particular, $\mathbb V=\mathbb{HSP}$.
		
	\end{enumerate}
\end{lem}

\begin{proof}
1, 2, 3.  Follow from the fact that any $c'$-reflexive homomorphism is
$c$-reflexive as well. 

4.  Let $\A\in\mathbb{SH}_c(\mathcal K)$.  Then, there exists
$\B\in\mathcal K$ and a surjective $c$-reflexive homomorphism
$f:\B\twoheadrightarrow\C$ such that $\A\leq\C$. 

We have that $\B'=f^{-1}(\A)\leq\B$, hence $\B'\in\mathbb S(\mathcal
K)$ and there exists the surjective homomorphism
$f|_{\B'}:\B'\twoheadrightarrow\A$.  Since $f$ is $c$-reflexive, also
$f|_{\B'}$ must be $c$-reflexive.  Hence, $\A\in\mathbb{H}_c\mathbb
S$. 

5.  Let $\A\in\mathbb{PH}_c(\mathcal K)$.  Then, there exist a family
$(\B_i)_{i\in I}\subseteq\mathcal K$ and a family of surjective
$c$-reflexive homomorphisms $f_i:\B_i\twoheadrightarrow\A_i$ such that
$\A=\pd_{i\in I}\A_i$.  But then, there exists a surjective
homomorphism $f:\pd_{i\in I}\B_i\twoheadrightarrow\pd_{i\in I}\A_i$
defined by $f(b)(i)=f_i(b(i))$.  Moreover, since each $f_i$ is
$c$-reflexive, also $f$ must be a $c$-reflexive.  Hence,
$\A\in\mathbb{H}_c\mathbb P(\mathcal K)$. 

6.  Let $\A\in\mathbb{PS}(\mathcal K)$.  Then, $\A=\pd_{i\in I}\A_i$
for some $\A_i\leq\B_i\in\mathcal K$.  But then, it is not difficult to
see that $\pd_{i\in I}\A_i\leq\pd_{i\in I}\B_i$ implying
$\A\in\mathbb{SP}(\mathcal K)$.   

7.  The class of $c$-reflexive homomorphisms is closed under
composition.  All these are trivial. 

8.  We obviously have $\mathbb{H}_c\mathbb
V_c=\mathbb{SV}_c=\mathbb{IPV}_c=\mathbb V_c$.  \\Since
$\mathbb{I\subseteq V}_c$, $\mathbb{H}_c\mathbb{SP\subseteq
H}_c\mathbb{SPV}_c=\mathbb V_c$. 

Since $\mathbb H_c$ is idempotent, $\mathbb{H}_c(\mathbb
H_c\mathbb{SP})=\mathbb H_c\mathbb{SP}$. 

Applying the previous results we get
$\mathbb{S(H}_c\mathbb{SP)\subseteq H}_c\mathbb{SSP=H}_c\mathbb{SP}$
and 
\\ $\mathbb{P(H}_c\mathbb{SP)\subseteq H}_c\mathbb{PSP\subseteq
H}_c\mathbb{SPP\subseteq H}_c\mathbb{SIPIP=H}_c\mathbb{SIP\subseteq
H}_c\mathbb{SH}_c\mathbb{P\subseteq H}_c\mathbb
H_c\mathbb{SP=H}_c\mathbb{SP}$. 

Obviously $\mathcal K\subseteq\mathbb{H}_c\mathbb{SP}(\mathcal K)$ and
$\mathbb{H}_c\mathbb{SP}(\mathcal K)$ is closed under $\mathbb{H}_c$,
$\mathbb{S,~ P}$.  Since $\mathbb V_c(\mathcal K)$ is the smallest
class containing $\mathcal K$ and closed under $\mathbb{H}_c$,
$\mathbb{S,~ P}$, we get that $\mathbb V_c(\mathcal
K)\subseteq\mathbb{H}_c\mathbb{SP}(\mathcal K)$. 
\end{proof}

\subsection{Algebraic Semantics for Conditional Quantitative Equations}

Quantitative algebras are used to interpret quantitative equational theories.

Given a quantitative algebra $\A=(A,\Omega^\A,d^\A)$ of type $\Omega$ and a
set $X$ of variables, an \emph{assignment} on $\A$ is an
$\Omega$-homomorphism $\alpha:\term X\to A$; it is used to interpret
abstract terms in $\T X$ as concrete elements in $\A$.  We denote by
$\term(X|\A)$ the set of assignments on $\A$.

\begin{defi}[Satisfiability]\label{sat}	
	A quantitative algebra $\A=(A,\Omega^\A,d^\A)$ of type $\Omega$ under the assignment $\alpha\in\T(X|\A)$ \emph{satisfies} a conditional quantitative equation  
	$\Gamma\vdash s=_\e t\in\E(\term X)$,  whenever   
	$$[d^\A(\alpha(t'),\alpha(s'))\leq \e'\mbox{ for all
	}s'=_{\e'}t'\in\Gamma]~\text{ implies }~d^\A(\alpha(s),\alpha(t))\leq\e.$$ 
	This is denoted by $$\Gamma\models_{\A,\alpha} s=_\e t.$$ 
	$\A$
	\emph{satisfies} $\Gamma\vdash
	s=_\e t\in\E(\term X)$, written $$\Gamma\models_\A s=_\e
	t,$$ if $\Gamma\models_{\A,\alpha} s=_\e
	t,$ for all assignments $\alpha\in\term(X|\A)$; in this case
	$\A$ is a \emph{model} of the conditional quantitative equation.
\end{defi}

Similarly, for a set $\U$ of conditional quantitative equations (e.g.,
a quantitative equational theory), we say that $\A$ is a model of $\U$
if $\A$ satisfies each conditional quantitative equation in $\U$.

If $\mathcal K$ is a class of quantitative algebras we write
$$\Gamma\models_\mathcal K s=_\e t,$$ if for any $\A\in\mathcal K$,
$\Gamma\models_\A s=_\e t$ .  Furthermore, if $\U$ is a quantitative
equational theory we write $$\mathcal K\models\U$$ if all algebras in
$\mathcal K$ are models for $\U$.

For the case of unconditional equations, note that the left-hand side of
the implication that defines the satisfiability relation in Definition \ref{sat} is vacuously satisfied.  For
these, instead of $\emptyset\models_{\A,\alpha} s=_\e t$ and
$\emptyset\models_\A s=_\e t$ we will often write
$\A,\alpha\models s=_\e t$ and $\A\models s=_\e t$ respectively.
Furthermore, for a class $\mathcal K$ of quantitative algebras, $$\mathcal K\models s=_\e t$$
denotes that $\A\models s=_\e t$ for all $\A\in\mathcal K$.

With these concepts in hand we can proceed and define equational classes.

\begin{defi}[Equational Class of Quantitative Algebras]
	For a signature $\Omega$ and a set $\U\subseteq\E(\T X)$ of conditional quantitative equations over the $\Omega$-terms $\term X$, the
	\emph{conditional equational class induced by $\U$} is the class of quantitative
	algebras of signature $\Omega$ satisfying $\U$.  
\end{defi}

We denote this class, as well as the full subcategory of
$\Omega$-quantitative algebras satisfying $\U$, by $\mathbb
K(\Omega,\U)$.  We say that a class of algebras that is a conditional equational 
class is \emph{conditional-equationally definable}.  

If $S$ is an axiomatization for $\U$, the equational class induced by $\U$ coincides with the equational class induced by $S$.

\begin{lem}\label{subalgebras}
	Given a set $\U$ of conditional quantitative equations of type $\Omega$ over $\T X$, $\KK(\Omega,\U)$ is closed under taking isomorphic images and subalgebras.  Consequently, if $\mathcal K$ is a class of quantitative algebras over $\Omega$, then $\mathcal K$, $\mathbb I(\mathcal K)$ and $\mathbb S(\mathcal K)$ satisfy the same conditional quantitative equations.  
\end{lem}

\begin{proof}
The closure w.r.t. isomorphic images derives trivially from Definition \ref{sat}. We prove now the closure under subalgebras.

Let $\A\in\KK(\Omega,\U)$ and $\B\leq\A$. We prove that $\B\in\KK(\Omega,\U)$.

Since $\B\leq\A$, $id_\B:B\to A$ defined by $id_\B(b)=b$ is a morphism of quantitative algebras.

Suppose that $\{s_i=_{\e_i}t_i\mid i\in I\}\vdash s=_e t\in\U$.  Hence, $$\{s_i=_{\e_i}t_i\mid i\in I\}\models_\A s=_e t\in\U,$$ meaning that for any $\alpha\in\term(X|\A)$, $$[d^\A(\alpha(s_i),\alpha(t_i))\leq\e_i~\text{ for all }~i\in I]~~\text{ implies }~~ d ^\A(\alpha(s),\alpha(t)\leq e.$$
Consider an arbitrary $\alpha\in\term(X|\B)$ and note that $\alpha\in\term(X|\A)$ as well.  

Suppose that $[d^\B(\alpha(s_i),\alpha(t_i))\leq\e_i~\text{ for all }~i\in I]$.  This is equivalent to $$[d^\A(\alpha(s_i),\alpha(t_i))\leq\e_i~\text{ for all }~i\in I].$$ But then, we also have $d^\A(\alpha(s),\alpha(t))\leq e$.  Hence, $d^\B(\alpha(s),\alpha(t))\leq e$.
\end{proof}


\section{The Variety Theorem for Basic Conditional Equations}

In this section we focus on the quantitative equational theories that
admit an axiomatization containing only basic conditional equations,
i.e., conditional equations of type $$\{x_i=_{\e_i}y_i\mid i\in I\}\vdash s=_\e t,$$
for $x_i,y_i\in X$, $s,t\in\term X$ and $\e_i,\e\in\prationals$.  We shall
call such a theory \textit{basic equational theory}.  

For a cardinal $c\leq\aleph_1$, a basic equational theory is a
\textit{$c$-basic equational theory} if it admits an axiomatization
containing only $c$-basic conditional equations, i.e., of
type $$\{x_i=_{\e_i}y_i\mid i\in I\}\vdash s=_\e t,$$ for $|I|<c$,
$x_i,y_i\in X$, $s,t\in\term X$ and $\e_i,\e\in\prationals$.  

An $\aleph_0$-basic equational theory is called a \textit{finitary-basic
	equational theory}; it admits an axiomatization containing only
finitary-basic conditional equations, i.e., of
type $$\{x_i=_{\e_i}y_i\mid i\in 1,..,n\}\vdash s=_\e t,$$ for
$n\in\naturals$, $x_i,y_i\in X$, $s,t\in\term X$ and
$\e_i,\e\in\prationals$.  

A $1$-basic equational theory is called an \textit{unconditional equational
	theory}; it admits an axiomatization containing only unconditional
equations of type $\emptyset\vdash s=_\e
t,$ for $s,t\in\term X$ and $\e\in\prationals$.  


\subsection{Closure under Products and Homomorphisms}

The basic equational theories are special since they guarantee, for their equational class, the closure under direct products, as the following lemma states.

\begin{lem}\label{products}
	If $\U$ is a basic equational theory (in particular, finitary-basic or unconditional), then $\KK(\Omega,\U)$ is closed under direct products.  
\end{lem}

\begin{proof}

Assume that $(\A_i)_{i\in I}\subseteq\KK(\Omega,\U)$.  We know that since $\U$ is a basic theory, it exists an axiomatization for $\U$ containing only basic quantitative equations.
It is sufficient to prove that whenever all $\A_i$ satisfy a basic quantitative equation, this is also satisfied by $\pd_{i\in I}\A_i$.

\[
\begin{tikzcd}
&\T X \arrow[ld,swap,"\alpha"] \arrow[d,"\pi_j\circ\alpha"]& \\
\pd_{i\in I}\A_i \arrow[r,"\pi_j"]&  \A_j
\end{tikzcd}
\]

Observe, for the begining, that for any assignment $\alpha:\term X\to\pd_{i\in I}\A_i$ and any $j\in I$,  $\pi_j\circ\alpha\in\T(X|\A_j)$ is an assignment in $\A_j$.

Consider now an arbitrary basic quantitative equation $$\{x_i=_{\e_i}y_i\mid i\in I\}\vdash s=_\e t\in\U.$$  Suppose that for all $i\in I$, $$\{x_i=_{\e_i}y_i\mid i\in I\}\models_{\A_i} s=_\e t.$$  This means that for any assignment, in particular for $\pi_j\circ\alpha$, we have $$d_j(\pi_j(\alpha(x_i)),\pi_j(\alpha(y_i)))\leq\e_i\text{ for all }i\in I~\text{ implies }~d_j(\pi_j(\alpha(s)),\pi_j(\alpha(t)))\leq\e.$$ 
Denote by $d$ the product metric and suppose that for an arbitrary assignment $\alpha\in\T(X|\pd_{i\in I}\A_i)$, $$d(\alpha(x_i),\alpha(y_i))\leq\e_i\text{ for all }i\in I.$$ Hence, 
$$\sup\{d_j(\pi_j(\alpha(x_i)),\pi_j(\alpha(y_i)))\mid\j\in I\}\leq\e_i\text{ for all }i\in I,$$ implying further that for each $j\in I$,
$$d_j(\pi_j(\alpha(x_i)),\pi_j(\alpha(y_i)))\leq\e_i\text{ for all }i\in I.$$
But then, the hypothesis guarantees that for any $j\in I$,	
$$d_j(\pi_j(\alpha(s)),\pi_j(\alpha(t)))\leq\e,$$ equivalent to
$$\sup\{d_j(\pi_j(\alpha(s)),\pi_j(\alpha(t)))\mid j\in I\}\leq\e.$$
Hence, 
$$d(\alpha(s),\alpha(t))\leq\e.$$
In conclusion, all the axioms of $\U$ (which are basic quantitative equations) must be satisfied by $\pd_{i\in I}\A_i$ implying that $\pd_{i\in I}\A_i\models\U$.
\end{proof}

The $c$-reflexive homomorphisms play a central role in characterizing the
basic equational theories in the case of the regular cardinals\footnote{
	The regular cardinals are the cardinals that cannot be obtained by using
	arithmetic involving smaller cardinals.  Thus, for example, $23$ is not a
	regular cardinal but $1$, $\aleph_0$ or $\aleph_1$ are, because none of
	them can be written as a smaller sum of smaller cardinals.}.  In fact,
because our signature admits only functions of countable (including finite)
arities, we will only focus on three regular cardinals: $1$, $\aleph_0$ and
$\aleph_1$.

The next lemma relates the classes of quantitative algebras that admit $c$-basic quantitative equational axiomatizations to their closure under $c$-reflexive homomorphisms.

\begin{lem}\label{homoimg}
	If $\U$ is a $c$-basic equational theory, where $c$ is a non-null regular
	cardinal, then $\KK(\Omega,\U)$ is closed under $c$-reflexive homomorphic images.  
	In particular,
	\begin{itemize}
		\item if $\U$ is an unconditional equational theory, then $\KK(\Omega,\U)$ is closed under homomorphic images;
		\item if $\U$ is a finitary-basic equational theory, then $\KK(\Omega,\U)$ is closed under $\aleph_0$-reflexive homomorphic images; 
		\item if $\U$ is a basic equational theory, then $\KK(\Omega,\U)$ is closed under
		$\aleph_1$-reflexive homomorphic images.  
	\end{itemize}
\end{lem}

\begin{proof}
Let $\A\in\KK(\Omega,\U)$, where $\U$ is a $c$-basic equational theory, $f:\A\rightarrow\B$ a $c$-reflexive homomorphism for $\B=f(\A)$.  Since $f$ is a homomorphism, $\B$ is a quantitative
algebra and $f:\A\twoheadrightarrow\B$ is obviously surjective.

Let $\{x_i=_{\e_i}y_i\mid i\in I\}\vdash s=_\e t\in\U$ be a $c$-basic
quantitative equation (hence $|I|< c$) satisfied by $\A$.  Assume that the terms
$s$ and $t$ depend on (a subset of) the set
$$\{x_i, y_i\mid i\in I\}\cup\{z_j\mid j\in J\}\subseteq X.$$ We denote
this by $s(\str{x}{i}{I},\str{y}{i}{I},\str{z}{j}{J})$ and
$t(\str{x}{i}{I},\str{y}{i}{I},\str{z}{j}{J})$.  Here, the $z_j$ are
variables that may occur in the terms $s,t$ but are not among the
variables that occur in the left-hand side of the basic inference.

Since
$$\{x_i=_{\e_i}y_i\mid i\in I\}\models_\A
s(\str{x}{i}{I},\str{y}{i}{I},\str{z}{j}{J})=_\e
t(\str{x}{i}{I},\str{y}{i}{I},\str{z}{j}{J}),$$
for any assignment $\alpha\in\term(X|\A)$,
$[d^\A(\alpha(x_i),\alpha(y_i))\leq \e_i~\text{ for all }~i\in I]$
implies
$$d^\A(s(\alpha(x_i))_{i\in I},(\alpha(y_i))_{i\in I},\alpha(z_j)_{j\in
J}),t((\alpha(x_i))_{i\in I},(\alpha(y_i))_{i\in I},\alpha(z_j)_{j\in
J})\leq\e.$$

Suppose there exists $\beta\in\term(X|\B)$ such that
$[d^\B(\beta(x_i),\beta(y_i))\leq \e_i~\text{ for all }~i\in I]$.  

Let $(a_i)_{i\in I}, (b_i)_{i\in I}, (c_j)_{j\in J}\subseteq\B$
s.t.  $\beta(x_i)=a_i$, $\beta(y_i)=b_i$ and $\beta(z_j)=c_j$.

Since $c$ is a regular cardinal, hence closed under union, $\{a_i,b_i\mid
i\in I\}\subseteq_c\B$.  Because $f$ is $c$-reflexive, there exist \\$m_i\in
f^{-1}(a_i),~ n_i\in f^{-1}(b_i)\in\A$ for each $i\in I$, such
that $$d^\A(m_i,n_i)=d^\B(a_i,b_i).$$ Since $f$ is surjective there exist
$u_j\in f^{-1}(c_j)$ for all $j\in J$.  

Let us write $s_A$ for the element of $\A$ obtained by substituting $m_i$
for the $x_i$, $n_i$ for the $y_i$ and $u_j$ for the $z_j$; similarly we
write $t_A$.  We write $s_B$ for the element of $\B$ obtained by
substituting $a_i$ for the $x_i$, $b_i$ for the $y_i$ and $c_j$ for the
$z_j$.  

Now the algebra $\A$ satisfies the basic quantitative equation, so using the
substitution that produces $s_A$ and $t_A$ we conclude that
$d^{\A}(s_A,t_A) \leq \epsilon$. 

 The homomorphism $f$ maps $s_A$ to $s_B$
and $t_A$ to $t_B$ and being non-expansive we conclude that
$d^{\B}(s_B,t_B)\leq \epsilon$.  

This proves that $\B$ also satisfies the basic quantitative equation.

\end{proof}

Putting together the results of Lemma \ref{subalgebras}, Lemma \ref{products} and Lemma \ref{homoimg}, we get the following result that emphasize the role of the $c$-basic quantitative equations for $c$-varieties.

\begin{cor}\label{satbasic}
	Let $\mathcal K$ be a class of quantitative algebras over the same
	signature and $c\leq\aleph_1$ a regular non-null cardinal.  Then $\mathcal K$,
	$\mathbb P(\mathcal K)$, 
	$\mathbb
	H_c(\mathcal K)$ and $\mathbb V_c(\mathcal K)$ satisfy all the same
	$c$-basic conditional equations.  
\end{cor}


\subsection{Canonical Model and Weak Universality}
In this subsection we give the quantitative analogue of the canonical model 
construction and prove weak universality.  Before we begin the detailed
arguments, we note a few points.  In the original variety theorem for universal algebras one
proceeds by looking at all congruences on the term algebra and quotienting
by the coarsest.  This strategy does not work in the present case.  We need
to consider the pseudometrics induced by all assignments of variables;
next, instead of quotienting by the kernel of the coarsest pseudometric, as
the analogy with the usual case would suggest, we need to take the
product of the quotient algebras indexed by these pseudometrics.  We note that this
is indeed a generalization of the non-quantitative case where, coincidentally, this
product algebra is isomorphic to the quotient algebra by the coarsest
congruence.  However, our proof here shows that the natural construction
that guarantees the weak universality, even when one considers reflexive
homomorphisms, is the product of the quotient algebras.  

Consider, as before, an algebraic similarity type $\Omega$ and a set $X$ of
variables. 
Let $\mathcal P_{\term X}$ be the set of all pseudometrics
$p:\term X^2\to\extreals$ such that all the functions in $\Omega$ are
non-expansive with respect to $p$.
For arbitrary $p\in\mathcal P_{\term X}$, let $$\term X|_p=(\term
X|_{ker(p)},\Omega,p)$$ be the quantitative algebra obtained by taking the
quotient of $\term X$ with respect to the congruence relation\footnote{The
	non-expansiveness of $p$ w.r.t.  all the functions in $\Omega$ guarantees
	that $ker(p)$ is a congruence with respect to $\Omega$.}
$$ker(p)=\{(s,t)\in\term X^2\mid p(s,t)=0\}.$$ 
Let $\mathcal K$ be a family  of quantitative algebras of type
$\Omega$ and $$\mathcal {P_K}=\{p\in\mathcal P_{\term X}\mid \term
X|_p\in\mathbb{IS}(\mathcal K)\}.$$ 
We begin by showing that $\mathcal{P_K}\neq\emptyset$ whenever
$\mathcal K\neq\emptyset$.

Consider an algebra $\A\in\mathcal K$, let $\alpha\in\T(X|\A)$ be an
arbitrary assignment and $[\alpha]:\T
X^2\to\extreals$ a pseudometric defined for arbitrary $s,t\in \T X$
by $$[\alpha](s,t)=\inf\{\e\mid\A,\alpha\models s=_\e t\}.$$  

\begin{lem}\label{alpha-pseudo}
	If $\A\in\mathcal K$ and $\alpha\in\T(X|\A)$, then
	$[\alpha]\in\mathcal{P_K}$.  Moreover, $\T X|_{[\alpha]}$ is a quantitative
	algebra isomorphic to $\alpha(\T X)$.  
\end{lem}

\begin{proof}
	The fact that $[\alpha]$ is a pseudometric follows directly from the
	algebraic semantics.  
	
	Let $f:|I|\in\Omega$ and $\str{s}{i}{I},\str{t}{i}{I}\subseteq\T X$.
	Assume that $[\alpha](s_i,t_i)\leq\e$ for all $i\in I$.  This means
	that for each $i\in I$, $\A,\alpha\models s_i=_\delta t_i$ for any
	$\delta\in\prationals$ with $\delta\geq\e$.  The soundness of (NExp)
	provides $\A,\alpha\models f(\str{s}{i}{I})=_\delta f(\str{t}{i}{I})$,
	i.e., $[\alpha](f(\str{s}{i}{I}),f(\str{t}{i}{I}))\leq\delta$ for any
	$\delta\geq\e$.  And this proves that $[\alpha]\in\mathcal{P}_{\T X}$.
	
	We know that $\alpha:\T X\to\A$ is a homomorphism of quantitative algebras,
	hence $\alpha(\T X)\leq\A$ and $\hat\alpha:\T X\to\alpha(\T X)$ defined by
	$\hat\alpha(s)=\alpha(s)$ for any $s\in\T X$ is a surjection.  Since from
	the way we have defined $[\alpha]$ we have
	that $$\hat\alpha(s)=\hat{\alpha}(t)~~\text{ iff }~~[\alpha](s,t)=0,$$ we
	obtain that the map $\ol{\alpha}:\T X|_{[\alpha]}\to\alpha(\T X)$ defined by
	$\ol{\alpha}(s|_{[\alpha]})=\alpha(s)$ for any $s\in\T X$, where
	$s|_{[\alpha]}$ denotes the $ker([\alpha])$-congruence class of $s$, is a QAs
	isomorphism.  
\end{proof}

The previous lemma states that for any algebra $\A\in\mathcal K$ and any
assignment $\alpha\in\T(X|\A)$,
$$\T X|_{[\alpha]}\simeq\alpha(\T X)\leq\A.$$
Since a consequence of it is $\mathcal{P_K}\neq\emptyset$ whenever
$\mathcal K\neq\emptyset$, we can define a pointwise supremum over the
elements in $\mathcal{P_K}$: 
$$d^{\mathcal K}(s,t)=\displaystyle\sup_{p\in\mathcal
	P_{\mathcal K}}p(s,t),\text{ for arbitrary }s,t\in\T X.$$ 
It is not difficult to notice that, $d^{\mathcal K}\in\mathcal
P_{\term X}$.

Let $\md X=(\pd_{p\in\mathcal{P_K}}\T X|_p,\Omega,d^\mathcal K)$ be the
product quantitative algebra with the index set $\mathcal{P_K}$.  

For arbitrary $s\in\T X$, let $\ang s\in\md X$ be the element such that for
any $p\in\mathcal{P_K}$, $\pi_p(\ang s)=s|_p$, where $s|_p\in\T X|_p$
denotes the $ker(p)$-equivalence class of $s$.   

Now note that, if $\mathcal K$ is a class  of quantitative algebras of the
same type containing non-degenerate elements, then the map $\gamma:\T
X\to\md X$ defined by $\gamma(t)=\ang t$ for any $t\in\T X$ is an injective
homomorphism of $\Omega$-algebras. 

\begin{lem}\label{l4}
If $\mathcal K$ is a non-trivial class of quantitative algebras of the
same type, the map $\gamma:\T X\to\md X$ defined by $\gamma(t)=\ang t$
for any $t\in\T X$ is an injective homomorphism of $\Omega$-algebras. 
\end{lem}

\begin{proof}
Since for any $p\in\mathcal{P_K}$, $ker(p)$ is a congruence and
$ker(d^\mathcal K)=\displaystyle\bigcap_{p\in\mathcal{P_K}}ker(p)$,
$\gamma$ is obviously a homomorphism of $\Omega$-algebras.   

We prove now that it is injective.  In order to have that for two
distinct terms $s,t\in\T X$ we have $\ang s=\ang t$, we need that for
any $p\in\mathcal{P_K}$, $s|_p=t|_p$.  Since for any
$p\in\mathcal{P_K}$, $ker(p)$ is a congruence, this will only happen if
there exist two distinct variables $x,y\in X$ such that
$\ang x=\ang y$. 

Note that if $p\in\mathcal{P_K}$ and $\sigma:X\to X$ is a bijection,
then $\sigma p$ defined by $\sigma p(s,t)=p(\sigma(s),\sigma(t))$ is an
element of $\mathcal{P_K}$ as well and $\T X|_p\simeq\T X|_{\sigma p}$.
With this observation we can conclude that if there exist two distinct
variables $x,y\in X$ such that $\ang x=\ang y$, then for any two
distinct variables $u,v\in X$ we have $\ang u=\ang v$, which implies
that $\mathcal K$ only contains degenerate algebras, a contradiction.
\end{proof}



In order to state now the weak universality property for a class $\mathcal K$ of quantitative algebras, we need firstly to identify a cardinal that plays a key role in our statement as an upper bound for the reflexive homomorphisms.  We shall denote it by $r(\mathcal K)$:
$$\begin{array}{ll}
r(\mathcal K)= & \left\{
\begin{array}{ll}
\aleph_1 & \textrm{if }\exists\A\in\mathcal K,~|\A|^+\geq\aleph_1\\
\sup\{|\A|^+\mid\A\in\mathcal K\} & \textrm{otherwise }\\
\end{array}\right.  \\
\end{array}$$ 
where $|\A|$ denotes of the cardinal of the support set of $\A$ and $c^+$
denotes the successor of the cardinal $c$.   

The following theorem is a central result of this paper.  One might be tempted to just use a quotient by $ker(d^{\KK})$ but in that case the
homomorphism that one gets by weak universality does not satisfy the
$c$-reflexive condition. 
\begin{thm}[Weak Universality]\label{weakuniv}
	Consider a class $\mathcal K$ of quantitative algebras containing non-degenerate elements.  For any $\A\in\mathcal K$ and any map $\alpha: X\to\A$ there exists a
	$r(\mathcal K)$-reflexive homomorphism $\beta:\md X\to\A$ such
	that $$\text{ for any }x\in X,~\beta(\ang x)=\alpha(x).$$
	
\end{thm}

\begin{proof}
	
	Let $\cat{QA_\Omega}$ be the category of $\Omega$-quantitative algebras.
	
	The map $\alpha:X\to\A$ can be canonically extended to an
	$\Omega$-homomorphism $\hat{\alpha}:\T X\to\A$. 
	
	Let $\gamma:\T X\hookrightarrow\md X$ be the aforementioned injective
	homomorphism of $\Omega$-algebras. 
	
	From Lemma \ref{alpha-pseudo} we know that  $\T
	X|_{[\hat\alpha]}\simeq\hat\alpha(\T X)\leq\A.$ So, we consider the
	projection $\pi_{[\hat\alpha]}:\md X\twoheadrightarrow\T X|_{[\hat\alpha]}$
	which is a surjective morphism of quantitaive algebras. 
	
	Let $\ol\alpha:\T X|_{[\hat\alpha]}\to\hat\alpha(\T X)$ be the
	isomorphism of quantitaive algebras defined in (the proof of)
	Lemma \ref{alpha-pseudo}.
	
	These maps give us the following commutative diagram.  
	
	\[
	\begin{tikzcd}
	&\text{in }\cat{Set}&&&\text{in }\cat{QA_\Omega}\\[-4ex]
	X \arrow[r, hook, "id_X"] \arrow[d, swap, "\alpha"] &
	\T X\arrow[r, hook, "\gamma"] \arrow[dl,swap,"\hat\alpha"] 
	& \md X\arrow[dll,swap,"\beta"]\arrow[d,two heads,"\pi_{[\hat\alpha]}"] & & \md X\arrow[d,"\beta"]\\
	\A  & \hat\alpha(\T X) \arrow[l,hook',"id_{\hat\alpha(\T X)}"] & \T X|_{[\hat\alpha]}\arrow[l,two heads,hook',"\ol\alpha"] & & \A
	\end{tikzcd}
	\]
	
	The diagonal of this diagram is a map $\beta$ defined for arbitrary
	$u\in\md X$ as
	follows: $$\beta(u)=\ol\alpha\circ\pi_{[\hat\alpha]}(u).$$ 
	
	Note that if $u=\ang s$ for some $s\in\T X$, then
	$$\beta(\ang s)=\ol\alpha(\pi_{[\hat\alpha]}(\ang
	s))=\ol\alpha(s|_{[\hat\alpha]})=\hat\alpha(s)$$
	and further more, if $x\in X$,
	$$\beta(\ang x)=\ol\alpha(\pi_{[\hat\alpha]}(\ang
	x))=\ol\alpha(x|_{[\hat\alpha]})=\hat\alpha(x)=\alpha(x).$$
	%
	Since $\beta$ is the composition of two homomorphisms of quantitative
	algebras, it is a homomorphism of quantitative algebras.  
	
	Finally we show that $\beta$ is a $r(\mathcal K)$-reflexive.  

	To start with, note that $\hat\alpha(\T X)\leq\A$ is the image of $\md
	X$
	through $\beta$.  Since $|\hat\alpha(\T X)|<r(\mathcal K)$, it only remains to
	prove that there exists a subset in $\md X$ such that for any
	$a,b\in\hat\alpha(\T X)$ we find two elements $u,v$ in this subset such
	that $\beta(u)=a$, $\beta(v)=b$ and $$d^\A(a,b)=d^\mathcal K(u,v).$$ 
	Let $s,t\in\T X$ be such that $\hat\alpha(s)=a$ and
	$\hat\alpha(t)=b$.  Let $u,v\in\md X$ such that
	$\pi_{[\hat\alpha]}(u)=s|_{[\hat\alpha]}$,
	$\pi_{[\hat\alpha]}(v)=t|_{[\hat\alpha]}$ and for any
	$p\neq [\hat\alpha]$, $\pi_p(u)=\pi_p(v)$.
	
	Since $d^\mathcal
	K(u,v)=\displaystyle\sup_{p\in\mathcal{P_K}}p(\pi_p(u),\pi_p(v))$ and
	$\pi_p(u)=\pi_p(v)$ for $p\neq[\hat\alpha]$, we obtain that
	indeed $$d^\mathcal K(u,v)=[\hat\alpha](s,t)=d^\A(a,b).$$
\end{proof}

Observe that the homomorphism $\beta$ is not unique, since any pseudometric $p\in\mathcal{P_K}$ can be associated to a projection $\pi_p$ that will eventually define a homomorphism of type $\beta$ making the diagram commutative -– hence, we have weak-universality. However, only for $\beta$ associated to $[\alpha]$, can we guarantee that $\beta$ is $r(\mathcal K)$-reflexive.

The weak universality reflects a fundamental relation between $\md X$ and
the $r(\mathcal K)$-reflexive closure operator $\mathbb
H_{r(\mathcal K)}$, as stated below. 

\begin{cor}\label{c01}
	If $\A\in\mathcal K$, then for $X$ sufficiently large, $$\A\in\mathbb
	H_{r(\mathcal K)}(\{\md X\}).$$  
\end{cor}

\begin{proof}
Let $X$ be a set such that $|X|\geq|\A|$.  Then, there exists a
surjective map $\alpha:X\twoheadrightarrow\A$.  Let $\beta:\md X\to\A$
be the $r(\mathcal K)$-reflexive homomorphism of quantitative
algebras defined in the previous theorem.  Since $\alpha$ is surjective,
so is $\beta$.  
\end{proof}

\begin{cor}\label{l01}
	Suppose that $\T X\neq\emptyset\neq\mathcal K$.  Then,
	$$\md X\in\mathbb{H}_{r(\mathcal K)}\mathbb{SP}(\mathcal K).$$
	Hence, if $\mathcal K$ is closed under $\mathbb{H}_{r(\mathcal K)}$, $\mathbb{S}$ and $\mathbb P$, then $\md X\in\mathcal K$.
\end{cor}

\begin{proof}
Note that there exists a map
$\alpha:X\to\pd_{p\in\mathcal{P_K}}\T X|_p$ defined by $\alpha(x)=\ang
x$.  
Then, applying the weak universality result, in Theorem \ref{weakuniv},
we get that there exists a $c$-reflexive homomorphism $\beta:\md
X\to\pd_{p\in\mathcal{P_K}}\T X|_p$.  
\end{proof}

The following theorem explains why we refer to $\md X$ as to the canonical model: it is because the class $\mathcal K$ and the quantitative algebra $\md X$ satisfy the same $c$-basic quantitative equations for any non-null regular
cardinal $c\leq r(\mathcal K)$.

\begin{thm}\label{t01}
	Let $\mathcal K$ be a class of quantitative algebras containing
	non-degenerate elements and $c\leq r(\mathcal K)$ a non-null regular
	cardinal.  Let $$\{x_i=_{\e_i}y_i\mid i\in I\}\vdash s=_\e t$$ be an
	arbitrary $c$-basic conditional equation on $\T X\neq\emptyset$, i.e., $|I|<c$.
	Then, $$\{x_i=_{\e_i}y_i\mid i\in I\}\models_{\mathcal K} s=_\e
	t~\text{ iff }~\{x_i=_{\e_i}y_i\mid i\in I\}\models_{\md X} s=_\e
	t.$$ 
\end{thm}

\begin{proof}
$(\Longrightarrow):$	If $\{x_i=_{\e_i}y_i\mid i\in I\}\models_{\mathcal
K} s=_\e t$, then Corollary \ref{satbasic} and Lemma \ref{operators}
guarantee that $\{x_i=_{\e_i}y_i\mid i\in
I\}\models_{\mathbb{H}_c\mathbb{SP}(\mathcal K)} s=_\e t$.  From Lemma
\ref{l01} applying also Lemma \ref{operators} we know that $\md
X\in\mathbb{H}_c\mathbb{SP}(\mathcal K)$.  Hence, $$\{x_i=_{\e_i}y_i\mid
i\in I\}\models_{\md X} s=_\e t.$$ 

$(\Longleftarrow):$	Suppose now that $\{x_i=_{\e_i}y_i\mid i\in
I\}\models_{\md X} s=_\e t$.  And assume in addition that $s$ and $t$ depend
on (a subset of) the set $$\{x_i,y_i\mid i\in I\}\cup\{z_j\mid j\in
J\}\subseteq X$$ of variables.  We denote this by writing, as before,
$s(\str{x}{i}{I}, \str{y}{i}{I}, \str{z}{j}{J})$ and $t(\str{x}{i}{I},
\str{y}{i}{I}, \str{z}{j}{J})$.  

Suppose there exists $\A\in\mathcal K$ such that $$\{x_i=_{\e_i}y_i\mid
i\in I\}\not\models_{\A} s(\str{x}{i}{I}, \str{y}{i}{I},
\str{z}{j}{J})=_\e t(\str{x}{i}{I}, \str{y}{i}{I}, \str{z}{j}{J}).$$ 
This means that there exists $\alpha\in\T(X|\A)$ such that $$\text{for
all }i\in I,~ d^\A(\alpha(x_i),\alpha(y_i))\leq\e_i~\text{ and
}~d^\A(\alpha(s), \alpha(t))>\e.$$ Moreover,
$\alpha(s)=s((\alpha(x_i))_{i\in I},(\alpha(y_i))_{i\in
I},(\alpha(z_j))_{j\in J})$ and \\$\alpha(t)=t((\alpha(x_i))_{i\in
I},(\alpha(y_i))_{i\in I},(\alpha(z_j))_{j\in J})$.  

Applying the weak universality, Theorem \ref{weakuniv}, we obtain that
$\alpha$ can be extended to a $r(\mathcal K)$-reflexive homomorphism
$\beta:\md X\to\A$ such that $\alpha(x)=\beta(\ang x)$ for any $x\in
X$.  Since $c\leq r(\mathcal K)$, $\beta$ is also $c$-reflexive.  

Since $c$ is regular, $|\{x_i,y_i\mid i\in I\}|< c$. 

Because $\alpha(x_i), \alpha(y_i)\in\hat\alpha(\T X)=\beta(\md X)\leq\A$ and $\beta$ is $c$-reflexive, we obtain that there exist $m_i,n_i\in\md X$ for all $i\in I$ such that
$\alpha(x_i)=\beta(m_i)$, $\alpha(y_i)=\beta(n_i)$ and $d^\A(\alpha(x_i),\alpha(y_{i}))=d^\mathcal K(m_i,n_{i'})$ for any $i\in I$.

Also $\alpha(z_j)\in\hat\alpha(\T X)=\beta(\md X)$, hence there exists $u_j\in\md X$ such that $\alpha(z_j)=\beta(u_j)$ for all $j\in J$.

From here we derive firstly that $$\text{for all }i\in I,~ d^\mathcal K(m_i,n_i)=d^\A(\alpha(x_i),\alpha(y_i))\leq\e_i.$$ Secondly, since $\beta$ is non-expansive, $$d^\mathcal K(s((m_i)_{i\in I},(n_i)_{i\in I},(u_j)_{j\in J}),t((m_i)_{i\in I},(n_i)_{i\in I},(u_j)_{j\in J})))\geq d^\A(\alpha(s),\alpha(t))>\e.$$

With these results in hand, we can define $\alpha_0\in\T(X|\md X)$ such that $$\alpha_0(x_i)=m_i,~ \alpha_0(y_i)=n_i~\text{ for any }i\in I\text{ and }\alpha_0(z_j)=u_j\text{ for any }j\in J.$$ The previous results demonstrates that 
$$\text{for all }i\in I,~ d^\mathcal K(\alpha_0(x_i),\alpha_0(y_i))\leq\e_i~\text{ and }~d^\mathcal K(\alpha_0(s), \alpha_0(t))>\e$$ which contradicts the fact that $\{x_i=_{\e_i}y_i\mid i\in I\}\models_{\md X} s=_\e t$.  
\end{proof}

This last result can further be instantiated for unconditional quantitative equations, which, in addition, can be used to characterize the metric $d^\mathcal K$.

\begin{cor}\label{t001}
	Let $\mathcal K$ be a class of quantitative algebras containing
	non-degenerate elements and $\T X\neq\emptyset$.  Then for arbitrary
	$s,t\in\T X$ and arbitrary $\e\in\prationals$, $$\mathcal K\models
	s=_\e t~\text{ iff }~\md X\models s=_\e t ~\text{ iff }~d^\mathcal
	K(\ang s,\ang t)\leq\e.$$ 
\end{cor}

\begin{proof}
The equivalence between the first two statements follows directly from
Theorem \ref{t01}.  

For the equivalence with the last statement, suppose that $\md X\models s=_\e t$.  Since the injection $\gamma:\T
X\hookrightarrow\md X\in\T(X|\md X)$ is an assignment, we obtain that
$d^\mathcal K(\gamma(s),\gamma(t))\leq\e$.  Hence, $d^\mathcal K(\ang
s,\ang t)\leq\e$.  

Suppose now that $d^\mathcal K(\ang s,\ang t)\leq\e$.  Then for any
$p\in\mathcal{P_K}$, $p(s,t)\leq\e$.  \\Let $\A\in\mathcal K$ and
assume that $s$ and $t$ depend of $(x_i)_{i\in I}\in X$; for
convenience we denote the two terms by $s(\str{x}{i}{I})$ and
$t(\str{x}{I}{I})$.

Consider arbitrary $\str{a}{i}{I}\subseteq\A$ and let
$\alpha\in\T(X|\A)$ such that $\alpha(x_i)=a_i$ for any $i\in I$.  

For arbitrary $i,j$ we have that $d^\mathcal K(\ang{x_i},\ang{x_j})\geq
d^\A(a_i,a_j)$ because as long as $\mathcal K\neq\emptyset\neq\T X$,
for any distinct variables $x,y\in X$, 
$$d^\mathcal K(\ang x,\ang y)=\sup_{\A\in\mathbb S(\mathcal K)}\sup_{a,b\in\A}d^\A(a,b).$$ 

Theorem \ref{weakuniv} guarantees that the aforementioned $\alpha$ can
be extended to a homomorphism $\beta:\md X\to\A$, which is
non-expansive.  Hence, $$d^\A(s(\str{a}{i}{I},),t(\str{a}{i}{I}))=d^\A(s((\alpha(x_i))_{i\in
I}),t((\alpha(x_i))_{i\in I}))=$$ $$d^\A(s((\beta(\ang{x_i})_{i\in
I}), t((\beta(\ang{x_i})_{i\in I}))\leq d^\mathcal K(\ang s,\ang
t)\leq\e.$$ 

Consequently, for any $\A\in\mathcal K$ and any assignment
$\alpha\in\T(X|\A)$, $$d^\A(\alpha(s),\alpha(t))\leq\e,$$ implying
$\mathcal K\models s=_\e t$.  
\end{proof}

\begin{cor}\label{c02}
	Let $\mathcal K\neq\emptyset\neq\T X$ and let $Y$ be a set of variables such that $|Y|\geq |X|$.  For any $c$-basic conditional equation $\{x_i=_{\e_i}y_i\mid i\in I\}\vdash s=_\e t$, where $c\leq r(\mathcal K)$ is a non-zero regular cardinal,  $$\{x_i=_{\e_i}y_i\mid i\in I\}\models_{\mathcal K} s=_\e t~\text{ iff }~\{x_i=_{\e_i}y_i\mid i\in I\}\models_{\md Y} s=_\e t.$$
\end{cor}


\subsection{Variety Theorem}

With these results in hand, we are ready to prove a general variety theorem
for quantitative algebras. 

Hereafter the signature $\Omega$ remains fixed; so, if $S$ is an
axiomatization for $\U$, we use $\KK(S)$ to denote the class
$\KK(\Omega,\U)$. 

If $S$ is a set of $c$-basic conditional equations, we say that $\KK(S)$ is
a \textit{$c$-basic conditional equational class}.  We call an
$\aleph_1$-basic conditional equational class simply \textit{basic
	equational class}.  A \textit{finitary-basic equational class} is an
$\aleph_0$-basic conditional equational class.  An \textit{unconditional
	equational class} is a $1$-basic conditional equational class.

We propose now a symmetric concept: if $\mathcal K$ is a set of
quantitative algebras and $0<c\leq\aleph_1$ is a cardinal, let $\E^c_X(\mathcal
K)$ be the set of all $c$-basic conditional equations over the set $X$ of
variables that are satisfied by all the elements of $\mathcal K$.   

\begin{lem}\label{l03}
	If $\mathcal K$ is a non-void $c$-variety for a regular cardinal
	$0<c\leq r(\mathcal K)$ and $X$ is an infinite set of variables,
	then $$\mathcal K=\KK(\E^c_X(\mathcal K)).$$
\end{lem}

\begin{proof}
	Let $\mathcal K'=\KK(\E^c_X(\mathcal K))$.  Obviously $\mathcal K\subseteq\mathcal K'$.
	
	We prove for the beginning that $\E^c_X(\mathcal K)=\E^c_X(\mathcal K')$.
	
	Since $\mathcal K\subseteq\mathcal K'$, $\E^c_X(\mathcal
	K)\supseteq\E^c_X(\mathcal K')$.  
	
	Let $\Gamma\vdash\phi\in\E^c_X(\mathcal K)$ be a $c$-basic quantitative
	inference.  
	Then, for any $\A\in\mathcal K$,
	$\Gamma\models_\A\phi$.  Consider an arbitrary $\B\in\mathcal
	K'$.  Since $\mathcal K'=\KK(\E^c_X(\mathcal K))$, $B$ must satisfy all
	the $c$-basic conditional equations in $\E^c_X(\mathcal K)$; in particular,
	$\Gamma\models_\B\phi$.  Hence, $\E^c_X(\mathcal
	K)\subseteq\E^c_X(\mathcal K')$.  
	
	Consider now an arbitrary $\A'\in\mathcal K'$.  
	
	From Corollary \ref{c01}, for a suitable set $Y$ of variables such that
	$|Y|\geq r(\mathcal K')$, we can define a surjection $\alpha:\T
	Y\twoheadrightarrow\A'$.  
	
	For arbitrary $s\in\T Y$, let $s|_{\mathcal
		K}\in\pd_{p\in\mathcal{P_K}}\T Y|_p$ be the element\footnote{Observe that $s|_{\mathcal K}$ has been denoted by $\langle s\rangle$ previously, when $\mathcal K$ was fixed. We change the notation here because we need to speak of such elements for various classes $\mathcal K, \mathcal K'$.} such that for any
	$p\in\mathcal{P_K}$, $\pi_p(s|_{\mathcal K})=s|_p$ and similarly
	$s|_{\mathcal K'}\in\pd_{p\in\mathcal{P_{K'}}}\T Y|_p$ be the element
	such that for any $p\in\mathcal{P_{K'}}$, $\pi_p(s|_{\mathcal
		K'})=s|_p$.  
	
	Theorem \ref{weakuniv} provides an injection $\gamma':\T
	Y\hookrightarrow\T_{\mathcal K'}Y$ defined by $\gamma'(s)=s|_{\mathcal
		K'}$ for any $s\in\T Y$; and a $r(\mathcal K')$-reflexive
	homomorphism $\beta':\T_{\mathcal K'}Y\twoheadrightarrow\A'$ which has
	the property that $\beta'(s|_{\mathcal K'})=\alpha(s)$.  Moreover,
	$\beta'$ is a surjection since $\alpha$ is.  
	
	Because $c\leq
	r(\mathcal K)\leq r(\mathcal K')$, $\beta'$ is also
	$r(\mathcal K)$-reflexive and $c$-reflexive.  Note now that also
	$\hat{\beta'}:\gamma'(\T Y)\to\A'$, which is defined by
	$\hat{\beta'}(u)=\beta'(u)$ for any $u\in\gamma'(\T Y)$, is a
	surjective $c$-reflexive homomorphism of quantitative algebras such
	that $\hat{\beta'}(s|_{\mathcal K'})=\alpha(s)$.  
	
	Similarly, there exists an injection $\gamma:\T Y\hookrightarrow\md Y$
	defined by $\gamma(s)=s|_{\mathcal K}$ for any $s\in\T Y$.
	
	Consider now the following two quantitative algebras 
	$$\T Y|_{d^{\mathcal K}}=(\T Y|_{ker(d^{\mathcal K})},\Omega, d^{\mathcal
		K})~\text{ and }$$ $$\T Y|_{d^{\mathcal K'}}=(\T Y|_{ker(d^{\mathcal
			K'})},\Omega, d^{\mathcal K'}).$$ 
	Note that the functions $\theta:\T Y|_{d^{\mathcal K}}\to\gamma(\T
	Y)$ defined by $\theta(s|_{d^{\mathcal K}})=\gamma(s)$ and $\theta':\T
	Y|_{d^{\mathcal K'}}\to\gamma'(\T Y)$ defined by
	$\theta'(s|_{d^{\mathcal K'}})=\gamma'(s)$ are isomorphisms of
	quantitative algebras.	
	\[
	\begin{tikzcd}
	\md Y \geq \T Y|_{d^{\mathcal K}}& \T Y \arrow[l, hook', swap, "\gamma"] \arrow[r, hook, "\gamma'"] \arrow[d,two heads, swap,"\alpha"] 
	& \T Y|_{d^{\mathcal K'}}\arrow[dl,swap, two heads, "\hat{\beta'}"]\arrow[r, hook, "id"] & \T_{\mathcal K'} Y\arrow[dll, two heads, "\beta'"]\\
	& \A' & 
	\end{tikzcd}
	\]
	Repeatedly applying Corollary \ref{t001} we get that for arbitrary $s,t\in\T Y$, $$d^\mathcal K(s|_{\mathcal K},t|_{\mathcal K})=0~\text{ iff }$$ $$\md Y\models s=_0
	t,\text{iff }$$ $$\mathcal K\models s=_0 t,~\text{ iff }$$ $$\emptyset\vdash s=_0 t\in\E^c_Y(\mathcal K)~\text{(since $\E^c_Y(\mathcal
	K) =\E^c_Y(\mathcal K')$)}, \text{ iff }$$
   $$\emptyset\vdash s=_0 t\in\E^c_Y(\mathcal K'), \text{ iff }$$
   $$\mathcal K'\models s=_0 t, \text{ iff }$$
	$$\T_{\mathcal K'} Y\models s=_0 t, \text{ iff }$$ $$d^{\mathcal K'}(s|_{\mathcal K'},t|_{\mathcal K'})=0.$$  
	
	Hence, $ker(d^{\mathcal K})=ker(d^{\mathcal K'})$ implying that $\T Y|_{d^{\mathcal K}}$ and $\T
	Y|_{d^{\mathcal K'}}$ are isomorphic $\Omega$-algebras.  
	
	Similarly, we can apply Corollary \ref{t001} for arbitrary $s,t\in\T Y$ and $\e\in\prationals$, as we did it before for $\e=0$, and obtain: 
	$$d^\mathcal K(s|_{\mathcal K},t|_{\mathcal K})\leq\e \text{ iff }$$ 
	$$\md Y\models s=_\e t, \text{ iff }$$
	$$\mathcal K\models s=_\e t, \text{ iff }$$
	$$\emptyset\vdash s=_\e t\in\E^c_Y(\mathcal K), \text{ iff }$$
	$$\emptyset\vdash s=_\e t\in\E^c_Y(\mathcal K'), \text{ iff }$$
	$$\mathcal K'\models s=_\e t, \text{ iff }$$
	$$\T_{\mathcal K'} Y\models s=_\e t, \text{ iff }$$
	$$d^{\mathcal K'}(s|_{\mathcal K'},t|_{\mathcal K'})\leq\e;$$ and since this is true for any
	$\e\in\prationals$, we obtain $$d^\mathcal K(s|_{\mathcal
		K},t|_{\mathcal K})=d^{\mathcal K'}(s|_{\mathcal K'},t|_{\mathcal
		K'}).$$

	Hence, $\T Y|_{d^{\mathcal K}}$ and $\T Y|_{d^{\mathcal K'}}$ are
	isomorphic quantitative algebras implying further that $\gamma(\T Y)$
	is isomorphic to $\gamma'(\T Y)$.  
	
	Now, since $\A'$ is the $c$-homomorphic image of $\gamma'(\T Y)$, it is
	also a $c$-homomorphic image of $\gamma(\T Y)$.  But $\gamma(T Y)\leq\md
	Y$ and since $\mathcal K$ is  a $c$-variety, from Lemma \ref{l01} we
	know that $\md Y\in\mathcal K$, hence $\gamma(\T Y)\in\mathcal K$.  
	
	Consequently, $\A'\in\mathbb{H}_c(\mathcal K)$ and since $\mathcal K$
	is a $c$-variety, $\A'\in\mathcal K$, from which we conclude
	$\mathcal{K'\subseteq K}$.  
\end{proof}

Now we prove the variety theorem for quantitative algebras.  

\begin{thm}[$c$-Variety Theorem]
	Let $\mathcal K$ be a class of quantitative algebras and $0<c\leq r(\mathcal K)$ a regular cardinal.  Then, $\mathcal K$ is a $c$-basic conditional equational class iff $\mathcal K$ is a $c$-variety.
	In particular, 
	\begin{enumerate}
		\item $\mathcal K$ is an unconditional equational class iff it is a variety;
		\item $\mathcal K$ is a finitary-basic equational class iff it is an $\aleph_0$-variety;
		\item $\mathcal K$ is a basic equational class iff it is an $\aleph_1$-variety.
	\end{enumerate}
\end{thm}

\begin{proof}
	($\Longrightarrow$): $\mathcal K=\KK(\U)$ for some set $\U$ of
	$c$-basic conditional equations.  Then, $\mathbb V_c(\mathcal
	K)\models \U$ implying further that $\mathbb V_c(\mathcal
	K)\subseteq\KK(\U)=\mathcal K$.  Hence, $\mathbb V_c(\mathcal
	K)=\mathcal K$.  
	
	($\Longleftarrow$): this is guaranteed by Lemma \ref{l03}.
\end{proof}

\textbf{Birkhoff Theorem in perspective.} Before concluding this section, we notice that our variety theorem also
generalizes the original Birkhoff theorem.  This is because any congruence
$\cong$ on an $\Omega$-algebra $\A$ can be seen as the kernel of the
pseudometric $p_{\cong}$ defined by $p_{\cong}(a,b)=0$ whenever $a\cong b$
and $p_{\cong}(a,b)=1$ otherwise.  The quotient algebra $\A|_{\cong}$ is a
quantitative algebra.  Any quantitative equational theory satisfied by
$\A|_{\cong}$ can be axiomatized by equations involving only $=_0$ and
$=_1$, since $0$ and $1$ are the only possible distances between its
elements.  However, this algebra also satisfies the equation $x=_1y$ for
any two variables $x$ and $y$, because $1$ is the diameter of its support.
Consequently, the only non-redundant equations satisfied by such an algebra
are of type $s=_0 t$, and these correspond to the equations
of the form $s=t$.


\section{The Quasivariety Theorem for General Conditional Equations}

In this section we study the axiomatizability of classes of quantitative
algebras that can be axiomatized by conditional quantitative equations, but
not necessarily by basic conditional quantitative equations.  Thus, we are
now looking for more relaxed types of axioms and consequently we will
identify more relaxed closure conditions.

We prove that a class $\mathcal K$ of $\Omega$-quantitative algebras admits
an axiomatization consisting of conditional quantitative equations,
whenever it is closed under isomorphisms, subalgebras and what we call
subreduced products.  A subreduced product is a quantitative subalgebra of
a (special type of) product of elements in $\mathcal K$; however, while
these products are always $\Omega$-algebras, they are not always
quantitative algebras.  This closure condition allow us to generalize the
classical quasivariety theorem that characterizes the classes of universal
algebras with an axiomatization consisting of Horn clauses.

It is not trivial to see that a $c$-variety is closed under these operators
for any regular cardinal $c>0$ and so our quasivariety theorem extends the
$c$-variety theorem presented in the previous section.  Indeed, all
isomorphisms are $c$-reflexive homomorphisms and since a $c$-variety is
closed under subalgebras and products, it must be closed under subreduced
products, as they are quantitative subalgebras of the product.

However, to achieve these results we had to involve and generalize concepts
and results from model theory of first-order structures.  This required us
to restrict ourselves to the signatures $\Omega$ containing only functions
of finite arity. 


\subsection{Preliminaries in Model Theory}

In this subsection we recall some basic concepts and results about the
model theory of first order structures. 

A \textit{first-order language} is a tuple $\lng=(\Omega,\R)$ where
$\Omega$ is an algebraic similarity type containing functions of finite
arity and $\R$ is a set of relation symbols of finite arity.   

A \textit{first-order structure} of type $\lng=(\Omega,\R)$ is a tuple
$\M=(M,\Omega^\M,\R^\M)$ where $(M,\Omega^\M)$ is an $\Omega$-algebra and
for any relation $R:i\in\R$, $R^\M\subseteq M^i$. 

A \textit{morphism of first-order structures} of type $\lng=(\Omega,\R)$ is
a map $$f:(M,\Omega^\M,\R^\M)\to(N,\Omega^\N,\R^\N)$$ that is a homomorphism
of $\Omega$-algebras such that for any relation $R:i\in\R$ and
$m_1,..m_i\in M$,
$$(m_1,..,m_i)\in \R^\M~\text{ iff }~(f(m_1),..,f(m_i))\in\R^\N.$$ 

$\M=(M,\Omega^\M,\R^\M)$ is a \textit{subobject} of
$\N=(N,\Omega^\N,\R^\N)$ if $(M,\Omega^\M)$ is an $\Omega$-subalgebra of
$(N,\Omega^\N)$ and for any $R:i\in\R$ and $m_1,..m_i\in M$,
$$(m_1,..,m_i)\in R^\M\text{ iff }(m_1,..,m_i)\in R^\N;$$ In this case we write $\M\leq\N$. 

\textbf{Equational First-Order Logic.} Given a first-order structure
$\lng=(\Omega,\R)$ and a set $X$ of variables, let $\T X$ be the set of
terms induced by $X$ over $\Omega$. 
The \textit{atomic formulas} of type $\lng=(\Omega,\R)$ over $X$ are
expressions of the form 
\begin{itemize}
	\item $s=t$ for $s,t\in\T X$;
	\item $R(s_1,..,s_k)$ for $R:k\in\R$ and $s_1,..,s_k\in\T X$.
\end{itemize}

The set $\lng X$ of first-order formulas of type $\lng$ over $X$ is the
smallest collection of formulas containing the atomic formulas and closed
under conjunction, negation and universal quantification $\forall x$ for
$x\in X$.  In addition we consider, as usual, all the Boolean operators and the
existential quantification.

If $\M$ is a structure of type $\lng$, let $\lng_\M$ be the first-order
language obtained by adding to $\lng$ the elements of $\M$ as constants. 

Given a first-order formula $\phi(x_1,..,x_{i},..x_k)$, in which $x_1,..,x_k\in X$ are all the free variables, we denote by $\phi(x_1,..,x_{i-1},m,x_{i+1},..x_k)$, as usual, the formula obtained by replacing all the free occurrences of $x_i$ by $m\in\M$.

\textbf{Satisfiability.} For a closed formula $\phi\in\lng_\M$, we define $\M\models\phi$ inductively on the structure of formulas as follows.
\begin{itemize}
	\item $\M\models s=t$ for $s,t\in T X$ containing no variables iff $s^\M=t^\M$.
	\item $\M\models R(s_1,..,s_k)$ for $R:k\in\R$ and $s_1,..,s_k\in\T X$ containing no variables iff $(s_1^\M,..,s_k^\M)\in R^\M$;
	\item $\M\models\phi\land\psi$ iff $\M\models\phi$ and $\M\models\psi$;
	\item $\M\models\lnot\phi$ iff $\M\not\models\phi$;
	\item $\M\models\forall x\phi(x)$ iff $\M\models\phi(m)$ for any $m\in\M$.
\end{itemize}
The semantics of the derived operators is standard.  The 
de Morgan laws give us semantically-equivalent prenex forms for
any first-order formula.

A first-order formula is an \textit{universal formula} if it is in prenex
form and all the quantifiers are universal. 

A \textit{Horn formula} has the following prenex form
$$Q_1x_1..Q_kx_k(\phi_1(x_1,..,x_k)\land..\land\phi_j(x_1,..,x_k)\to\phi(x_1,..,x_k)),$$ where each $Q_l$ is a quantifier and each $\phi_l$ and $\phi$ is an atomic formula with (a subset of) the set $\{x_1,..,x_k\}$ of free variables\footnote{Some authors define a Horn formula as a conjunction of such constructs, or allow $\phi=\top$; none of these choices affect our development here.}.  

A \textit{universal Horn formula} is a Horn formula which is also an universal formula.

\textbf{Direct Products.} Given a non-empty indexed family $(\M_i)_{i\in I}$ of first-order structures of type $\lng=(\Omega,\R)$, where $\M_i=(M_i,\Omega^{\M_i},\R^{\M_i})$, the \textit{direct product} $\M=\pd_{i\in I}\M_i$ is the $\lng$-structure whose universe is the product set $\pd_{i\in I} M_i$ and its functions and relations are defined as follows, where $\pi_i:\pd_{i\in I}\M_i\to\A_i$ denotes the $i$-th projection.
\begin{itemize}
	\item for $f:k\in\Omega$, and $m_1,..,m_k\in \pd_{i\in I}M_i$, $$\pi_i(f^\M(m_1,..,m_k))=f^{\M_i}(\pi_i(m_1),..,\pi_i(m_k));$$
	\item for $R:k\in\R$, and $m_1,..,m_k\in \pd_{i\in I}M_i$,
	$$(m_1,..,m_k)\in R^\M~\text{ iff }~(\pi_i(m_1),..,\pi_i(m_k))\in R^{\M_i}~\text{ for all }~i\in I.$$
\end{itemize}



\textbf{Reduced Products.} Let $\str{\M}{i}{I}$ be an indexed family of first-order structures of type $\lng=(\Omega,\R)$ and $F$ a proper filter over $I$.

Consider the relation $\mathord\sim_F\subseteq\pd_{i\in I}\M_i\times\pd_{i\in I}\M_i$ s.t.  $$m\sim_F n~\text{ iff }~\{i\in I\mid \pi_i(m)=\pi_i(n)\}\in F.$$

It is known that  when $F$ is a proper filter of $I$, $\sim_F$ is a congruence relation with respect to the algebraic structure of $\M=\pd_{i\in I}\M_i$ (see, e.g., \cite[Lemma~2.2]{Burris81}).
This allows us to define the \textit{reduced product induced by a proper filter} $F$, written $(\pd_{i\in I}\M_i)|_F$, as the $\lng$ first-order structure such that 
\begin{itemize}
	\item its universe is the set $(\pd_{i\in I}M_i)|_{\sim_F}$, which is the quotient of $\pd_{i\in I}M_i$ with respect to $\sim_F$; we denote by $m_F$ the $\sim_F$-congruence class of $m\in \pd_{i\in I}M_i$;
	\item for $f:k\in\Omega$, and $(m^1,..,m^k)\in \pd_{i\in I}M_i$, $$f(m^1_F,..,m^k_F)=(f(m^1,..,m^k))_F:$$
	\item for $R:k\in\R$, and $(m^1,..,m^k)\in \pd_{i\in I}M_i$, 
	$$R(m^1_F,..,m^k_F)~\text{ iff }~\{i\in I\mid R(\pi_i(m^1),..,\pi_i(m^k)\}\in F.$$
\end{itemize}






\textbf{Quasivariety Theorem.} A class $\mathfrak M$ of $\lng$-structures is an \textit{elementary class} if there exists a set $\Phi$ of first-order $\lng$-formulas such that for any $\lng$-structure $\M$, $$\M\in\mathfrak M~\text{ iff }~\M\models\Phi.$$
An elementary class is an \textit{universal class} if it can be axiomatized by universal formulas; it is an \textit{universal Horn class} if it can be axiomatized by universal Horn formulas.

We conclude this section with the quasivariety theorem (see, e.g., \cite[Theorem~2.23]{Burris81}).  To state it, we define a few closure operators on classes of $\lng$-structures.

Let $\mathfrak M$ be an arbitrary class of $\lng$-structures.
\begin{itemize}
	\item $\mathbb I(\mathfrak M)$ denotes the closure of $\mathfrak M$ under isomorphisms of $\lng$-structures;
	\item $\mathbb S(\mathfrak M)$ denotes the closure of $\mathfrak M$ under subobjects of $\lng$-structures;
	\item $\mathbb P(\mathfrak M)$ denotes the closure of $\mathfrak M$ under direct products of $\lng$-structures;
	\item $\mathbb P_R(\mathfrak M)$ denotes the closure of $\mathfrak M$ under reduced products of $\lng$-structures.
\end{itemize} 

\begin{thm}[Quasivariety Theorem]\label{mt:quasi}
	Let $\mathfrak M$ be a class of $\lng$-structures.  The following statements are equivalent.
	\begin{enumerate}
		\item  $\mathfrak M$ is a universal Horn class;
		\item  $\mathfrak M$ is closed under $\mathbb{I,~S}$ and $\mathbb P_R$;
		\item  $\mathfrak M=\mathbb{ISP}_R(\mathfrak M')$ for some class $\mathfrak M'$ of $\lng$-structures.
	\end{enumerate}
\end{thm}


\subsection{Quantitative First-Order Structures}

In this subsection we identify a class of first-order structures, the \textit{quantitative first-order structures} (QFOs), which are the first-order counterparts of the quantitative algebras.

Given a first-order structure $\M=(M,\Omega^\M,\R^\M)$ of type $(\Omega,\R)$, $f:k\in\Omega$ and $R:l\in\R$, let $f(R^\M)\subseteq M^l$ be the set of the tuples $(f(m_1^1,..,m_k^1),..,f(m_1^l,..,m_k^l))$ such that for each $i=1,..,k$, $(m_i^1,..,m_i^l)\in R^\M$.

\begin{defi}\label{qfo} [Quantitative First-Order Structure]
	An $\Omega$-\textit{quantitative first-order structure} for a signature $\Omega$ is a first-order structure $\M=(M,\Omega^\M,\equiv^\M)$ of type $(\Omega,\equiv)$, where $\mathord\equiv=\{=_\e\mid\e\in\prationals\},$ that satisfies the following axioms for any $\e,\delta\in\prationals$
	\begin{enumerate}
		\item $\mathord{=_0^\M}$ is the identity on $\M$;
		\item $\mathord{=_\e^\M}$ is symmetric;
		\item $\mathord{=_\e^\M\circ=_\delta^\M\subseteq =_{\e+\delta}^\M}$;
		\item $\mathord{=_\e^\M\subseteq =_{\e+\delta}^\M}$;
		\item for any $f:k\in\Omega$, $\mathord{f(=^\M_\e)\subseteq =^\M_\e}$; 
		\item for any $\delta$, $\mathord{\displaystyle\bigcap_{\e>\delta}=_\e\subseteq=_\delta}$;
	\end{enumerate}
\end{defi}

\begin{thm}\label{correspondence}
	(i) Any quantitative algebra $\A=(A,\Omega,d)$ defines uniquely a quantitative first-order structure by $$a=_\e b~\text{ iff }~d(a,b)\leq\e.$$
	(ii) Any quantitative first-order structure $\M=(M,\Omega^\M,\equiv^\M)$ defines uniquely a quantitative algebra by letting $$d(m,n)=\inf\{\e\in\prationals\mid m=_\e n\}.$$
	These define an isomorphism between the category of $\Omega$-quantitative algebras and $\Omega$-quantitative first-order structures.
\end{thm}

\begin{proof}
The proof is trivial and relies on the fact that conditions (1)-(6) in Definition \ref{qfo} corresponds to (Refl), (Symm), (Triang), (Max), (Arch) and (NExp) respectively.
\end{proof}

Let $\cat{QA_\Omega}$ be the category of $\Omega$-quantitative algebras and $\cat{QFO_\Omega}$ the category of $\Omega$-quantitative first-order structures.  Theorem \ref{correspondence} defines two functors $\mathbb F$ and $\mathbb G$ that act as identities on morphisms, which define an isomorphism of categories as in the figure below.
\begin{equation*}
\begin{tikzpicture}[baseline={(m.center)},arrow label/.style={font=\scriptsize}]
\matrix (m) [commutative diagram={1cm}{1.2cm}] {
	\cat{QA_\Omega} & \cat{QFO_\Omega} \\
};
\path[-latex,arrow label]
(m-1-1) edge[bend left] node[above] {$\mathbb F$} (m-1-2)
(m-1-2) edge[bend left] node[above] {$\mathbb G$} (m-1-1)
;
\end{tikzpicture}
\end{equation*}
We already know that the subobjects and the direct products of quantitative first-order structures are first-order structures.  However, since the isomorphisms of categories preserve limits and colimits, we can prove that the subobjects and the direct products of quantitative first-order structures are, in fact, quantitative first-order structures, i.e., they satisfy the axioms (1)-(6) of Definition \ref{qfo}, as the next lemma establishes.

\begin{lem}\label{subobjandprod}
	I.  If $\M,\N$ are $\Omega$-quantitative first-order structures s.t.  $\M\leq\N$, then $$\mathbb G\M\leq\mathbb G\N.$$	
	II.  If $\str{\M}{i}{I}$ is a family of $\Omega$-quantitative first-order structures, then $$\mathbb G(\pd_{i\in I}\M_i)=\pd_{i\in I}\mathbb G\M_i.$$
	III.  If $\A,\B$ are $\Omega$-quantitative algebras such that $\A\leq\B$, then $$\mathbb F(\A)\leq\mathbb F\B.$$
	IV.  If $\str{\A}{i}{I}$ is a family of $\Omega$-quantitative algebras, then $$\mathbb F(\pd_{i\in I}\A_i)=\pd_{i\in I}\mathbb F\A_i.$$
\end{lem}

\subsection{Subreduced Products of Quantitative First-Order Structures}

Given an indexed family $\str{\M}{i}{I}$ of $\Omega$-quantitative first-order structures and a proper filter $F$ on $I$, we can construct, as before, the reduced product $(\str{\M}{i}{I})|_F$ of first-order structures, which is a first-order structure.  But it is not guaranteed that it satisfies the axioms in Definition \ref{qfo}.  From the definition of the reduced product we obtain a first-order structure $(\str{\M}{i}{I})|_F$ that enjoys the following property for any $\e\in\prationals$.
$$m_F=_\e n_F~\text{ iff }~\{i\in I\mid \pi_i(m)=_\e\pi_i(n)\}\in F.$$
Note that if for all $i\in I$, $\M_i$ satisfies the axioms (1)-(5) from Definition \ref{qfo}, then $(\str{\M}{i}{I})|_F$ satisfies them as well.

For instance, we can verify the condition (3): suppose that $m_F=_\e n_F$ and $n_F=_\delta u_F$.  Hence, $$\{i\in I\mid \pi_i(m)=_\e\pi_i(n)\}, \{i\in I\mid \pi_i(n)=_\delta\pi_i(u)\}\in F.$$ 
Since $F$ is a filter, it is closed under intersection, so 
$$\{i\in I\mid \pi_i(m)=_\e\pi_i(n)\text{ and }\pi_i(n)=_\delta\pi_i(u)\}\in F.$$ 
Now, axiom (3) guarantees that 
$$\{i\in I\mid \pi_i(m)=_\e\pi_i(n)\text{ and }\pi_i(n)=_\delta\pi_i(u)\}$$$$\subseteq\{i\in I\mid \pi_i(m)=_{\e+\delta}\pi_i(u)\}$$ and since $F$ is closed under supersets, $$\{i\in I\mid \pi_i(m)=_{\e+\delta}\pi_i(u)\}\in F.$$
Similarly, one can verify each of the axioms but (6).  This is because axiom (6) requires that any reduced product has the property that for any $\delta\in\prationals$, 
$$\{i\in I\mid \pi_i(m)=_\e\pi_i(n)\}\in F\text{ for all }\e>\delta$$ implies $$\{i\in I\mid \pi_i(m)=_\delta\pi_i(n)\}\in F.$$
This is a very strong condition not necessarily satisfied by a filter or an ultrafilter.  It is, for instance, satisfied by the filters and ultrafilters closed under countable intersections, but the existence of such filters requires measurable cardinals (see for instance \cite{Chang92} for a detailed discussion).

Hence, while the reduced products of quantitative first-order structures can always be defined as first-order structures, they are not always quantitative first-order structures, since they might not satisfy axiom (6) in Definition \ref{qfo}.  Therefore, taking reduced products and ultraproducts are not internal operations over the class of quantitative first-order structures of the same type, even if they are internal operations over the larger class of first-order structures of the same type.  This observation motivates our next definition.

\begin{defi}[Subreduced Products] 
	Given an indexed family $\str{\M}{i}{I}$ of quantitative first-order structures and a proper filter $F$ on $I$, a \textit{subreduced product} of this family induced by $F$ is any subobject $\M$ of the first-order structure $(\pd_{i\in I}\M_i)|_F$ such that $\M$ is a quantitative first-order structure.  
\end{defi}

Given a class $\mathfrak M$ of quantitative first-order structures of the same type, the closure of $\mathfrak M$ under subreduced products is denoted by $\mathbb P_{SR}(\mathfrak M)$.  

With this concept in hand we can generalize the quasivariety theorem for first-order structures to get a similar result for classes of QFOs that can be properly axiomatized.

\begin{thm}[Quasivariety Theorem for Quantitative First-Order Structures]\label{quasivar}
	Let $\mathfrak M$ be a class of $\Omega$-quantitative first-order structures.  Then, the following statements are equivalent.
	\begin{enumerate}
		\item $\mathfrak M$ is an universal Horn class;
		\item $\mathfrak M$ is closed under $\mathbb {I, S}$ and $\mathbb P_{SR}$;
		\item $\mathfrak M=\mathbb{ISP}_{SR}(\mathfrak M_0)$ for some class $\mathfrak M_0$ of $\Omega$-quantitative first-order structures.
	\end{enumerate}
\end{thm}

\begin{proof}
	$(1)\Longrightarrow(2)$: let $\mathfrak M$ be an universal Horn class of $\Omega$-QFOs.  Then there exists an universal Horn class of $\Omega$-first-order structures $\mathfrak M'$ that satisfies the same first-order theory $\mathcal T$ that $\mathfrak M$ does.  If we denote the class of $\Omega$-quantitative first-order theories by $\cat{QFO_\Omega}$, we have $$\mathfrak M=\mathfrak M'\cap \cat{QFO_\Omega}.$$
	Applying Theorem \ref{mt:quasi}, $\mathfrak M'$ is closed under $\mathbb{I,~S}$ and $\mathbb P_R$.
	
	Obviously, $\mathfrak M$ is closed under $\mathbb I$, since isomorphic first-order structures satisfy the same first-order sentences.  $\mathfrak M$ is also closed under $\mathbb S$, as Lemma \ref{subobjandprod} guarantees.
	
	Let $\{\M_i\mid i\in I\}\subseteq\mathfrak M$ and $F$ a proper filter of $I$.  
	
	Let $\M\leq(\pd_{i\in I}\M_i)|_F$ such that $\M\in \cat{QFO_\Omega}$.  
	
	Since $\{\M_i\mid i\in I\}\subseteq\mathfrak M'$ and $\mathbb P_R(\mathfrak M')=\mathfrak M'$, we get that $(\pd_{i\in I}\M_i)|_F\in\mathfrak M'$.  Hence, $\M\in\mathbb S(\mathfrak M')=\mathfrak M'$.  And further, $\M\in\mathfrak M'\cap \cat{QFO_\Omega}=\mathfrak M$.  In conclusion, $\mathfrak M$ is also closed under $\mathbb P_{SR}$.
	
	$(2)\Longrightarrow(3)$: since $\mathfrak M$ is closed under $\mathbb{I,~S}$ and $\mathbb P_{SR}$, $$\mathfrak M=\mathbb{ISP}_{SR}(\mathfrak M).$$
	$(3)\Longrightarrow(1)$: suppose that $\mathfrak M=\mathbb{ISP}_{SR}(\mathfrak M_0)$ for some class $\mathfrak M_0$ of quantitative first-order structures.  \\Let $\mathfrak M'=\mathbb{ISP}_{R}(\mathfrak M)$.  Applying Theorem \ref{mt:quasi}, $\mathfrak M'$ is a universal Horn class of first-order structures.  We prove now that $\mathfrak M=\mathfrak M'\cap \cat{QFO_\Omega}$.
	
	Let $\M\in\mathfrak M'\cap \cat{QFO_\Omega}$.  Then, $\M$ is isomorphic to some $\N\leq(\pd_{i\in I}\M_i)|_F$ for some $\str{\M}{i}{I}\subseteq\mathfrak M$ and a proper filter $F$ of $I$, and $\N\in \cat{QFO_\Omega}$.  Hence, $\M\in\mathbb{ISP}_{SR}(\mathfrak M)=\mathfrak M$.  And this concludes that $\mathfrak M'\cap \cat{QFO_\Omega}\subseteq\mathfrak M$.
	
	Since we have trivially  $\mathfrak M\subseteq\mathfrak M'\cap \cat{QFO_\Omega}$ from the way we constructed $\mathfrak M'$, we get that $\mathfrak M=\mathfrak M'\cap \cat{QFO_\Omega}$.
	
	Now, since $\mathfrak M'$ is a universal Horn class of first-order structures, we obtain that $\mathfrak M$ is a universal Horn class of quantitative first-order structures.
\end{proof}


\subsection{Subreduced Products of Quantitative Algebras}

Theorem \ref{quasivar} characterizes classes of $\Omega$-QFOs as universal Horn classes.  In this subsection we convert this result into a result regarding the axiomatizability of classes of quantitative algebras.


For the beginning, we note an equivalence between the conditional equations interpreted over the class of quantitative algebras and the universal Horn formulas interpreted over the class of quantitative first-order structures.  This relies on the fact that a quantitative equation of type $s=_\e t$ is also an atomic formula in the corresponding quantitative first-order language and vice versa.  The following theorem establishes this correspondence.

\begin{thm}\label{logicalequiv}
	Let $\phi_1(x_1,..,x_k)\ldots,\phi_l(x_1,..,x_k)$ and $\psi(x_1,..,x_k)$ be $\Omega$-quantitative first-order atomic formulas depending of the variables $x_1,..,x_k\in X$.  
	
	I.  If $\M$ is an $\Omega$-quantitative first-order structure, then the following statements are equivalent  $$\M\models\forall x_1..\forall x_k(\phi_1(x_1,..x_k)\land..\land\phi_l(x_1,..x_k)\to\psi(x_1,..x_k)),$$ 
	$$\{\phi_1(x_1,..,x_k)\land..\land\phi_l(x_1,..,x_k)\}\models_{\mathbb G \M}\psi(x_1,..,x_k).$$
	
	II.  If $\A$ is an $\Omega$-quantitative algebra, then the following statements are equivalent  $$\{\phi_1(x_1,..,x_k)\land..\land\phi_l(x_1,..,x_k)\}\models_{\A}\psi(x_1,..,x_k),$$
	$$\mathbb F\A\models\forall x_1..\forall x_k(\phi_1(x_1,..x_k)\land..\land\phi_n(x_1,..x_k)\to\psi(x_1,..x_k)).$$ 
\end{thm} 

As in the case of quantitative first-order structures, the concept of subdirect product of an indexed family of quantitative algebras for a given proper filter is not always defined.  The following definition reflects this issue.

\begin{defi}[Subreduced products of Quantitative Algebras]
	Let $\str{\A}{i}{I}$ be an indexed family of $\Omega$-quantitative algebras and $F$ a proper filter of $I$.  A \textit{subreduced product} of this family induced by $F$ is a quantitative algebra $\A$ s.t. $$\mathbb F\A\leq\pd_{i\in I}(\mathbb F\A_i)|_F.$$	
\end{defi}

Let $\mathbb P_{SR}(\mathcal K)$ be the closure of the class $\mathcal K$ of quantitative algebras under subreduced products.  Now we can provide the analogue of Theorem \ref{quasivar} for quantitative algebras as a direct consequence of Theorem \ref{correspondence} , Theorem \ref{quasivar} and Theorem \ref{logicalequiv}.

\begin{thm}\label{QA:quasivar}
	Let $\mathcal K$ be a class of $\Omega$-quantitative algebras.  The following statements are equivalent.
	\begin{enumerate}
		\item $\mathcal K$ is a conditional equational class;
		\item $\mathcal K$ is closed under $\mathbb{I,~S}$ and $\mathbb P_{SR}$;
		\item $\mathcal K=\mathbb{ISP}_{SR}(\mathcal K_0)$ for some class $\mathcal K_0$ of $\Omega$-quantitative algebras.
	\end{enumerate}
\end{thm} 



\subsection{Going further: Complete Quantitative Algebras}

The proof pattern that we developed to prove the quasivariety theorem for
QFOs, Theorem \ref{quasivar}, is actually more general and it could be used
to provide similar theorems for other classes of quantitative algebras.  In
\cite{Mardare16} we have shown that the class of quantitative algebras
defined over complete metric spaces plays a central in the theory of
quantitative algebras.  For this reason we will briefly show how a
quasivariety theorem could be done for complete metric spaces.

We call a quantitative algebra over a complete metric space a
\textit{complete quantitative algebra.}

If we follow the intuition behind Theorem \ref{correspondence}, we will
discover that we can define the concept of complete quantitative
first-order structure as being a quantitative first-order structure for
which the corresponding quantitative algebra through the functor $\mathbb
G$ is a complete quantitative algebra.  In fact, the completeness condition
can be encoded by an infinitary axiom to be added to the conditions (1)-(6)
in Definition \ref{qfo}, namely the axiom that requires that any Cauchy
sequence has a limit.  Let us call it the \textit{Cauchy condition}. 

We will be then able to prove that the category of $\Omega$-complete
quantitative algebras is isomorphic to the category of $\Omega$-complete
quantitative first-order structures. 

Further we can define, given a class $\mathbb M$ of $\Omega$-complete
quantitative first-order structures, the concept of complete-subreduced
product: given an indexed family $\str{\M}{i}{I}$ of $\Omega$-complete
quantitative first-order structures, a complete-subreduced product is any
$\Omega$-complete quantitative first-order structure that is a subobject of
the reduced product $\pd_{i\in I}\M_i|_F$ for some proper filter $F$ of
$I$. 

With this in hand, one can redo the proof of Theorem \ref{quasivar} in
these new settings and should obtain a quasivariety theorem for complete
QFOs.


\section{Conclusions}

In this paper we have established the fundamental results on the
axiomatizability of classes of quantitative algebras by equations,
conditional equations and Horn clauses.  These results required substantial
new techniques.  We have not put this work into a fully categorical
framework such as described in~\cite{Barr94,Adamek98,Adamek10,Manes12}.  We
are actively working on understanding these
connections and also the connections with enriched Lawvere theories. There 	is also much to understand when looking at other approaches to quantitative 	reasoning, for example the work of Jacobs and his group~\cite{Cho15}.



\section*{Acknowledgements}
Prakash Panangaden is supported by the Natural Science and Engineering
Research Council of Canada, Radu Mardare is supported by the Project
4181-00360 
of the Danish Council for Independent Research.  We thank
Giorgio Bacci, Florence Clerc, Robert Furber and Dexter Kozen for useful discussions.  We thank the Simons Institute
for hosting the Fall 2016  program on Logical Structures in Computation 
which we all attended at various times, benefitting 
from the
stimulating atmosphere there.

\bibliographystyle{alpha}
\bibliography{main}

\end{document}